\RequirePackage{amsthm}

\documentclass[aps,11pt,sn-mathphys-num]{sn-jnl}% Math and Physical Sciences Numbered Reference Style 
\usepackage{graphicx}%
\usepackage{multirow}%
\usepackage{amsmath,amssymb,amsfonts}%
\usepackage{amsthm}%
\usepackage{dsfont}
\usepackage{mathrsfs}%
\usepackage[title]{appendix}%
\usepackage{xcolor}%
\usepackage{textcomp}%
\usepackage{manyfoot}%
\usepackage{grffile}
\usepackage{booktabs}%
\usepackage{algorithm}%
\usepackage{algorithmicx}%
\usepackage{algpseudocode}%
\usepackage{listings}%
\usepackage[justification=justified]{caption}
\usepackage{comment}
\usepackage{accents}
\usepackage{enumitem} % For custom labels
\usepackage{txfonts}
\usepackage{cleveref}

\crefname{section}{Section}{Sections}
\crefname{subsection}{Subsection}{Subsections}
\crefname{appendix}{Appendix}{Appendices}

\crefname{theorem}{Theorem}{Theorems}
\crefname{proposition}{Proposition}{Propositions}
\crefname{corollary}{Corollary}{Corollaries}
\crefname{lemma}{Lemma}{Lemmas}
\crefname{definition}{Definition}{Definitions}
\crefname{table}{Table}{Tables}
\crefname{assumption}{Assumption}{Assumptions}

\theoremstyle{thmstyleone}%
\newtheorem{theorem}{Theorem}[section]% meant for sectionwise numbers
\newtheorem{proposition}[theorem]{Proposition}% 
\newtheorem{lemma}[theorem]{Lemma}% 
\theoremstyle{thmstyletwo}%
\newtheorem{remark}[theorem]{Remark}
\newtheorem{assumption}[theorem]{Assumption}

\theoremstyle{thmstylethree}%
\newtheorem{definition}[theorem]{Definition}%
\numberwithin{equation}{section}

\def\eps{\varepsilon}
\newcommand{\gf}{{\Gamma}}

\newcommand{\gmc}{\mu}

\newcommand{\ind}{\mathds{1}}

\newcommand{\mL}{\mathscr{L}}
\newcommand{\mA}{\mathcal{A}}

\newcommand{\mC}{\mathcal{C}}
\newcommand{\C}{\mathcal{C}}
\newcommand{\B}{\mathcal{B}}
\newcommand{\R}{\mathbb{R}}

\renewcommand{\P}{\mathbb{P}}

\newcommand{\mW}{\mathcal{W}}

\newcommand{\E}{{\mathbb E}}
\newcommand{\T}{{\mathcal T}}
\newcommand{\be}{\begin{equation}}
\newcommand{\ee}{\end{equation}}

\newcommand{\widesim}{\sim}

\newcommand{\cb}{\color{black}}

\graphicspath{{figs/}}
\begin{document}
\title
    {\begin{center}
        Transport in multifractal Kraichnan flows:\\
	from turbulence to Liouville quantum gravity
    \end{center}
    }

\author*[1,2]{\fnm{André} L. P. \sur{Considera}}\email{andre.luis@impa.br}
\author*[3]{\fnm{Simon} \sur{Thalabard}}\email{simon.thalabard@univ-cotedazur.fr}
\affil*[1]{\orgname{Instituto de Matemática Pura e Aplicada -- IMPA}, 
\orgaddress{\state{Rio de Janeiro}, \country{Brazil}}} 

\affil[2]{
    %\orgdiv{SPEC/IRAMIC/DSM/CEA},
    \orgname{SPEC, CEA, CNRS, Université Paris-Saclay}, 
    \orgaddress{\street{Gif-sur-Yvette}, 
    \city{Paris}, \postcode{91191}, \country{France}}}

\affil*[3]{\orgname{Institut de Physique de Nice, Université Côte d’Azur CNRS - UMR 7010}, \orgaddress{\street{17 rue Julien Lauprêtre}, \city{Nice}, \postcode{06200}, \country{France}}}

\abstract{
We investigate the behavior of fluid trajectories in a multifractal extension of the Kraichnan model of turbulent advection. The model  couples  a one-dimensional, Gaussian, white-in-time random flow to a frozen-in-time Gaussian multiplicative chaos (GMC). The resulting velocity field features an interplay between the roughness exponent $\xi\in(0,2]$, controlling the correlation decay for the Gaussian component, and the intermittency parameter $\gamma \in [0,\sqrt {2}/2)$,  prescribing the deviations from self-similarity.  Recent numerical work by the authors suggests that such coupling induce a smoothing-by-intermittency effect, and the purpose here is to address this phenomenon theoretically. Using the theory of 1D Feller Markov processes,  
we characterize the phases  of the two-particle separation process upon varying $\xi$ and $\gamma$,  
extending to a multifractal setting   the stochastic/deterministic and colliding/non-colliding  
transitions known in the monofractal Kraichnan case. Our analysis distinguishes between two settings:  quenched  or annealed.  In the quenched setting, the GMC realization is prescribed, and we show 
that the phases are governed by the most probable Hölder exponent of the multifractal velocity field. In the annealed setting, the GMC is averaged over, leading to an additional smoothing effect. Moreover, we show that the separation process exhibits structural analogies with multiplicative one-dimensional versions of  the Liouville Brownian motion --- a diffusion process evolving in a random GMC   landscape, originally  introduced in the context of Liouville quantum gravity. In particular,  both the quenched and annealed phase transitions are recovered by considering a multiplicative LBM characterized by  a roughness parameter $\xi + 4\gamma^2$ and an intermittency exponent $\gamma$.
}

\keywords{Turbulence, Spontaneous stochasticity, 
Kraichnan model, 
Liouville quantum gravity,
Gaussian multiplicative chaos, 
Liouville Brownian motion}

\graphicspath{{figs/}}
\maketitle
\newpage
\tableofcontents
\newpage
%%%%%%%%%%%%%%%%%%%%%%%%%%%%%%%%%%%%%%%%%
%%%%%%%%%%%%%%%%%%%% 1: INTRODUCTION
%%%%%%%%%%%%%%%%%%%%%%%%%%%%%%%%%%%%%%%%%
\section{Introduction}\label{sec:intro}

The probabilistic description of Lagrangian dispersion by 
turbulent flows has its roots in the  work of Taylor 
and Richardson \cite{ taylor1922diffusion, richardson1926atmospheric}, 
who promoted the idea that multiscale carrier flows give rise to 
non-trivial transport properties. 
For example, it is now commonly accepted that in  
homogeneous isotropic incompressible turbulence,  
fluid trajectories separate explosively, 
with their mean square distance undergoing  
a phase of anomalous diffusive growth $\simeq t^{3}$, 
independent of the initial separation: This result
is known as Richardson's law 
\cite{jullien1999richardson,boffetta2002relative,biferale2005lagrangian,bitane2012time,bitane2013geometry,thalabard2014turbulent,bourgoin2015turbulent}.
On the other hand, the Lagrangian flow  
behaves quite differently in  compressible environments. 
The paradigmatic  examples are Burgers flows, where fluid particles collide and eventually coalesce 
into  single trajectories, leading to shock formation 
and concentration of mass into singular structures
\cite{frisch2002burgulence,bec2007burgers}.

These two contrasting behaviors —-- explosive separation in incompressible 
turbulence and coalescence in compressible Burgers flows --— can both be recovered,  
along with a wide range of intermediate regimes, in transport models  
driven by synthetic random velocity ensembles. 
In this setting, the velocity field is prescribed statistically 
rather than derived from fluid equations, allowing one to systematically 
explore how different statistical features of the flow influence 
Lagrangian transport. 
Solvable frameworks are  obtained when considering
$N$ tracer particles  subjected to thermal noise, evolving  in a quasi-Lagrangian frame 
\cite{chaves2003lagrangian,kupiainen2003nondeterministic,gawedzki2008soluble}.  Their collective motion relative to a reference trajectory $X_0$ is  then
defined by the 
stochastic differential equation (SDE) 
\begin{equation} \label{eq:ql_transport}
    dX_p = U(X_p-X_0,dt)+\sqrt{2\kappa}\, d\beta_p,
\end{equation} 
where $p=1, \dots, N$ denotes the label of the Lagrangian particle,
$U$ is a random field with prescribed statistics, and
$\{\beta_p\}_{p=1}^N$ is a family of $d$-dimensional
Brownian motions that are independent from one another and from $U$.
The parameter $\kappa > 0$ represents the amplitude of thermal noise.
One may identify different phases of the Lagrangian flow,
based on various dichotomic behaviors \cite{chaves2003lagrangian},    
by examining the separation process 
$R(t) := X_1(t) - X_0(t)$
and the associated transition function  
\be \label{eq:dicho1}
P_{t}^\kappa(r_0,dr) := \P_{r_0}\left(R(t) \in dr\right),
\ee  
obtained by averaging over realizations  
of both the velocity field and the thermal noise.
Of particular interest is the dichotomy revealed by the limiting property
\begin{equation}
	\label{eq:dicho2}
   \lim_{\kappa,r_0 \to 0}  P_t^\kappa(r_0,dr)\underset{(=)}{\neq}
   \delta(r)dr,
\end{equation}
 distinguishing between stochastic ($\neq$) and deterministic ($=$) 
behavior of trajectories evolving in quenched space-time 
realizations of the velocity $U$, and differing only by vanishingly 
small thermal diffusivity.
For $\kappa>0$, the presence of thermal noise causes 
the separation process to spread out, 
even when starting the process at zero distance.
The phenomenon of \emph{spontaneous stochasticity} --- named by analogy with 
spontaneous symmetry breaking in spin systems --- 
corresponds to the case where stochasticity remains even 
in the limit $\kappa \to 0$. 
In this regime, initially coincident fluid particles separate and reach 
$O(1)$ distances in finite time, forming a well-defined stochastic process 
already in typical space-time realizations of the velocity field 
and in the absence of thermal noise.
This signals a breakdown of the Lagrangian flow, with significant implications 
for the physics and mathematics of scalar 
transport phenomena 
\cite{gawedzki2008soluble,eyink2015spontaneous,drivas2017lagrangian,valade2023anomalous}. 
Another dichotomy refers to the ability of particles,  
initially separated by $r_0 > 0$, to collide and intersect their 
trajectories in finite time.  
This is captured by the behavior of the transition probability  
to zero separation
\begin{equation} \label{eq:dichotomy}
    \lim_{\kappa \to 0}  P_t^\kappa(r_0, \{0\})\underset{(=)}{>} 0,
\end{equation}
which distinguishes between colliding ($>$) 
and non-colliding ($=$) behavior.\\

The dichotomies described by the limiting properties  
\eqref{eq:dicho2} and \eqref{eq:dichotomy}  
characterize two independent facets of Lagrangian dynamics.  
Investigating which combinations of these regimes arise  
under different flow statistics yields a 
phase diagram characterizing turbulent transport.
The phase diagram can be rigorously analyzed in the case of the 
Kraichnan model of turbulent advection, 
which considers Gaussian velocity ensembles with non-trivial spatial structure
but white-in-time statistics 
\cite{kraichnan1968small,bernard1998slow,frisch1998intermittency, gawedzki2000phase,falkovich2001particles,le2002integration,kupiainen2003nondeterministic,le2004flows,gawkedzki2004sticky,gabrielli2008clustering}.
In the Kraichnan model, the possible behaviors of the Lagrangian flow  
depend on the interplay between key parameters  
of the Gaussian field,  
such as its spatial dimensionality, level of compressibility,  
and spatial scaling properties.
In its one-dimensional version, one considers the 
stochastic dynamics \eqref{eq:ql_transport} 
with the driving velocity field given by the integral formula \cite{robert2008hydrodynamic}
\begin{equation} \label{eq:u0}
    U_{0}(x, dt) = 
    \int_{\mathbb{R}} \phi(x-y)\,  \mW(dy,dt),
\end{equation}
where $\mW$ is a $(1+1)$-dimensional Brownian sheet
\cite{khoshnevisan2006multiparameter} 
(see also \cite[Exercise I.3.11]{revuz2013continuous}), 
formally expressed as integral of  
space-time white noise. 
The kernel $\phi$ is a suitably defined convolution square root 
of the spatial part of the Kraichnan correlation function,
\begin{equation}
    \E \left(U_0(x,t)U_0(x+r,t')\right) = \C(r)\min (t,t'), 
\end{equation}
satisfying the following asymptotic relation at small distances
\begin{equation}
    \label{eq:correlation}
    \phi * \phi = \C(r) \sim 1-|r|^\xi \quad\text{as} \quad r \to 0.
\end{equation}
Here, the parameter $\xi\in (0,2]$ represents the roughness exponent 
of the underlying Kraichnan flow,
ensuring that typical realizations of the velocity field $U_0$ 
are only Hölder continuous with exponent $H$ for any $H < \xi/2$.

In the one-dimensional Kraichnan model,
the phase transitions of the Lagrangian flow
are governed, to a large extent, by the roughness exponent $\xi$.
For smooth flows $(\xi = 2)$, the velocity field $U_0$ is Lipschitz continuous,
ensuring existence and uniqueness of solutions to the SDE \eqref{eq:ql_transport}.
As a result, trajectories initialized at the same point cannot branch,
and distinct trajectories cannot intersect.
This corresponds to equality in both dichotomies \eqref{eq:dicho2} 
and \eqref{eq:dichotomy}, signaling deterministic, 
non-colliding Lagrangian behavior.
For $\xi < 2$, the flow becomes rough, and a wider range of behaviors  
can be observed.  
If $1 \leq \xi < 2$, the Lagrangian flow remains deterministic  
but exhibits collisions.  
This represents the collapse of the Lagrangian flow,  
with fluid particles coalescing into a single trajectory  
upon impact.
On the other hand, for $0 < \xi < 1$,  
different trajectories still intersect,  
but whether they are able to branch out or not  
depends on the specific details of the regularization employed.  
This marks a non-universal phase,  
in which the limiting behavior of the Lagrangian flow  
is sensitive to the regularization procedure,  
allowing for both spontaneous stochasticity and deterministic behavior  
to emerge under different conditions.
These transitions can be determined by observing that,  
as a stochastic process, the relative separation maps to a Bessel process  
in effective dimension  
$d_e = \frac{2 - 2\xi}{2 - \xi}$  
\cite{gawedzki2000phase,gawedzki2002soluble,lawler2018notes}.  
More fundamentally, they arise as a direct consequence of  
Feller's theory for one-dimensional diffusions  
\cite{breiman1992probability,revuz2013continuous}.

Although convenient proxies for turbulent flows,  
self-similar velocity ensembles are not realistic  
in light of the intermittency phenomenon and Kolmogorov's 1962 (K62)  
refined self-similarity hypothesis,  
which relates the lack of self-similarity in turbulent environments  
to the emergence of scale-dependent lognormal statistics  
for the energy dissipation rate \cite{kolmogorov1962refinement}.
An explicit construction of intermittent (multifractal) random velocity fields  
accounting for these features was first proposed by Robert and Vargas (2008)  
\cite{robert2008hydrodynamic} and was subsequently refined in  
\cite{chevillard2010stochastic,chevillard2015peinture,pereira2016dissipative,chevillard2019skewed,reneuve2020flow}.
The common building block for these models is the
Gaussian multiplicative chaos (GMC) theory \cite{rhodes2014gaussian},  
a theory of random multifractal measures initiated in the pioneering work of Kahane  
\cite{kahane1985chaos}, which provides a rigorous random field construction  
of the Kolmogorov-Obukhov lognormal model of turbulent energy dissipation.\\

In this work, we introduce a multifractal extension of the Kraichnan model  
by coupling the Kraichnan velocity field $U_0$ with a frozen-in-time GMC.
More specifically, we consider velocity fields given by 
the modified integral formula
\begin{equation}\label{eq:velocity}
    U(x, dt) = 
   Z^{-1/2} \int_{\mathbb{R}} \phi(x-y)\, {e^{\gamma \gf (y)}} \mW(dy,dt),\;\;  Z:=e^{2\gamma^2 \mathbb E\left({\gf^2}\right)}.
\end{equation}
In this modified expression, the parameter $\gamma \in (0,\sqrt{2}/2)$ 
quantifies the level of intermittency, and 
$\gf$ is a log-correlated Gaussian field independent of $\mW$. 
The normalization factor $Z$ compensates 
the wild fluctuations of the exponential term $e^{\gamma\gf}$.
Expression \eqref{eq:velocity} defines a statistically 
homogeneous and isotropic velocity field, whose equal-time structure functions satisfy
the lognormal scaling \cite{chevillard2015peinture,chevillard2019skewed,reneuve2020flow}
\begin{equation}
   \label{eq:uscaling}
    \mathbb E\left(U(x+r,t)-U(x,t)\right)^{2p}
    \underset{r\to 0}{\sim} |r|^{\zeta_U(2p)},
\end{equation}
with scaling exponents given by
\begin{equation}\label{eq:zeta_U}
    \zeta_U(p)=\left(\frac{\xi}{2}+\gamma^2\right)p-\frac{\gamma^2p^2}{2},
\end{equation}
valid for $p<\frac{\xi}{2\gamma^2} + 1$.
The non-linear power-law spectrum 
captures one facet of the intermittency phenomenon,
making the field \eqref{eq:velocity}
a more faithful stochastic caricature of turbulence.
Note that, in this work, the fields $\gf$ and $\mW$ 
are assumed to be independent from one another. 
As a result, the odd moments vanish identically. 
While intermittent in the sense of its multiscaling properties,
the multifractal velocity ensemble considered in this work is
in particular not skewed, unlike the more realistic field constructed in \cite{chevillard2019skewed}.

Our purpose  is to investigate the phase transitions  
of the Lagrangian flow in such multifractal extension of the Kraichnan model,  
where fluid particles are advected by the   intermittent velocity field \eqref{eq:velocity}.  
Compared to the monofractal case, the term $e^{\gamma\Gamma}$ introduces yet another layer of randomness. 
Describing the separation process thus requires  
distinguishing between two statistical settings:  
the \emph{quenched setting}, in which the statistical behavior is studied  
in a pathwise, almost sure sense with respect to realizations of  
$\gf$ and $\mW$, and the \emph{annealed setting},  
where one  averages over both $\mW$
and the environmental noise $e^{\gamma\Gamma}$.  
In both settings,  a phase diagram for the Lagranian flows can be derived  in terms of  
the scaling exponent $\xi$ and the intermittency parameter $\gamma$, 
using tools from one-dimensional Feller Markov processes.
Our analysis indicates that Eulerian intermittent fluctuations have a smoothing effect on 
particle dispersion. This smoothing effect holds  true in both the quenched and the annealed settings. 
In the quenched case, we find that the behavior of the Lagrangian flow has a clear physical relation to the multifractal properties of the velocity field.
The phases are governed by the most probable Hölder exponent $\overline{H}= \xi/2 + \gamma^2$ --- in the sense of Parisi-Frisch multifractal formalism --- reflecting the fact that fluid trajectories spend most of their time in regions of space exhibiting typical velocity fluctuations.
As for the annealed setting, our results recover a mean-field argument  put forward in the previous numerical study \cite{considera2023spontaneous},
where  the smoothing effect of intermittency was observed in Monte-Carlo simulations of the separation process. 
As such,  the present work serves as a theoretical counterpart to \cite{considera2023spontaneous}.  
We point out, however, a slight difference in the approach:  
In~\cite{considera2023spontaneous}, expression~\eqref{eq:velocity} was used to model  
the Eulerian velocity field. In the present work, by contrast,  
expression~\eqref{eq:velocity} models a quasi-Lagrangian velocity,  
and this seemingly minor alteration proves fundamental in providing a solvable framework.\\

The separation process driven by our  multifractal Kraichnan flows resembles a diffusion evolving in
a random geometry induced by the GMC. 
The use of GMC as a random landscape, 
in which other random processes are allowed to evolve, 
has been explored previously in \cite{garban2016liouville,berestycki2015diffusion}, 
within the context of Liouville quantum gravity (LQG). 
In those works, a two-dimensional Brownian motion is coupled to the GMC, 
giving rise to the natural diffusion process in the random geometry 
of planar LQG, known as Liouville Brownian motion (LBM). 
In this work, we explicitly introduce a multiplicative generalization  
of the Liouville Brownian motion in one dimension,  
which we refer to as the \emph{multiplicative Liouville Brownian motion} (MLBM),  
and compare it with the multifractal Kraichnan diffusion  
from the viewpoints of phase transitions and multifractal behavior.

Specifically, the $(\xi,\gamma)$-MLBM is formally defined as the solution to the  
stochastic dynamics  
\begin{equation}\label{eq:MLBM1}
    d\B = |\B|^{\xi/2} e^{-\gamma\gf(\B)  
    + \gamma^2 \E(\gf(\B)^2)} d\beta.
\end{equation}
Setting $\xi = 0$  
recovers a one-dimensional version of the Liouville Brownian motion  
\cite{garban2016liouville}, while setting $\gamma = 0$ yields a  
process with the same statistics as the inter-particle separation in  
the (monofractal) Kraichnan model.  
The logarithmic singularity in the correlation function of $\gf$ at  
small distances renders the expression above non-rigorous. To  
rigorously define the MLBM, we therefore follow \citet{garban2016liouville} and  
construct it via a time-change method.
Although they exhibit different phase transitions for a given  
pair $(\xi, \gamma)$, the phase transitions of the multifractal Kraichnan model  
can be derived from those of the MLBM  
by performing a suitable change of parameters.
By applying the mapping $\xi \mapsto \xi + 4\gamma^2$  
in \eqref{eq:MLBM1}, one recovers exactly the phase transitions  
of the multifractal Kraichnan model, in both the quenched and annealed settings.
Conversely,  
the transformation $\xi \mapsto \xi - 4\gamma^2$  
maps the multifractal Kraichnan model into the MLBM  
in a similar fashion.

Ending this introduction with a conceptual note,
GMC is a theory whose inception was rooted in the K62 theory of turbulence, 
but it has since expanded significantly to other 
areas of physics, most notably to LQG.  
In the present work, 
the comparison between the MLBM and the multifractal Kraichnan process  
illustrates how mathematical tools from LQG, based on GMC theory, 
can be applied to the study of turbulence when properly adapted.
It also provides a natural way to describe 
turbulent transport using the language of random geometries,
extending to the Lagrangian setting a mathematical bridge 
--- mediated by GMC theory --- 
between LQG and statistical theories of turbulence.\\
%This suggests that mathematical tools from LQG, based on GMC theory, can be applied  to the study of turbulence when properly adapted. 

The paper is organized as follows.  
Section \ref{sec:Kraichnan} recalls the 1D (monofractal) Kraichnan theory  
for the separation process and its connection to Feller's 
theory of one-dimensional diffusions. 
The definition of the separation process 
in multifractal Kraichnan flows is given in Section \ref{sec:direct}.
Section \ref{sec:mfk-phases}
discusses the colliding/non-colliding and stochastic/deterministic dichotomies  
in terms of the parameters $\xi$ and $\gamma$.  
The multifractal properties of the separation process are discussed 
in Section \ref{sec:multifractal}, 
while Section \ref{sec:MLBM} explores analogies with 
Liouville Brownian motion.
Section \ref{sec:remarks} formulates concluding remarks.

%%%%%%%%%%%%%%%%%%%%%%%%%%%%%%%%%%%%%%%%%
%%%%%%%%%%%%%%%  2: KRAICHNAN THROUGH SPEED AND SCALES
%%%%%%%%%%%%%%%%%%%%%%%%%%%%%%%%%%%%%%%%%
\section{Phase transitions in the 1D Kraichnan model} \label{sec:Kraichnan}

In this section, we describe the  phases of the Lagrangian flow in the one-dimensional  
Kraichnan model, in terms  of the roughness of the velocity field and regularization procedures.
The results presented here are not new, and are likely familiar to experts in turbulence theory. Our goal, however, is to rederive them  
using a set of mathematical tools --- namely, Feller's theory for one-dimensional diffusions, speed measures, and time-changed Brownian motions \cite{breiman1992probability,revuz2013continuous}
--- that may be less familiar to the mathematical physicist.  This background section sets up the machinery that  
will be used later in  Sections  \ref{sec:direct} and \ref{sec:mfk-phases} to rigorously analyze the phase transitions arising 
in the multifractal Kraichnan model ---  the main contribution of this work.

\subsection{The regularized separation process}
\subsubsection{Characterization as a diffusion process}

We begin by considering the regularized Lagrangian dynamics,  
in which particle motion is smoothed by the effects of viscosity  
and thermal diffusivity.
This is achieved by introducing a  
regularized version of the filtering kernel appearing in the velocity field  
\eqref{eq:u0}.
For clarity and simplicity, 
we follow \citet{robert2008hydrodynamic} (RV) and explicitly define the
regularized kernel by the expression
\begin{equation}
     \label{eq:rvkernel}
    \phi_\eta(x)=L^{-\xi/2}\psi\left(\dfrac{x}{L}\right)\frac{x}{|x|_{\eta}^{3/2-\xi/2}}.
\end{equation}
Here, $|\cdot|_\eta$ denotes a regularized norm with 
small-scale cutoff $\eta$, such that $|r|_\eta=|r|$ for $|r| >\eta$, 
and smoothly varies down to $|0|_\eta=O(\eta)$. 
The function $\psi(x)$ is a smooth, radially symmetric function equal  
to $1$ when $|x| \leq 1$ and vanishing when $|x| > 2$. It plays the role  
of a large-scale cutoff, localizing the field at scales smaller than $L$.  
The parameter $L > 0$ represents the (large-scale) velocity correlation length, with  
the prefactor $L^{-\xi/2}$ ensuring that the field is dimensionless.  
The regularized Kraichnan field $U_{0,\eta}$ is then obtained
by 
replacing $\phi$ with $\phi_\eta$ in expression \eqref{eq:u0}.

In what follows, we consider the dynamics \eqref{eq:ql_transport}
with $U_{0,\eta}$ as the driving field.
In this case, the two-particle separation $R(t):=X_1(t)-X_0(t)$
satisfies
\begin{equation}\label{eq:ql_model2}
    dR=  \Delta U_{0,\eta}(R,dt)+ \sqrt{2\kappa}\, d\beta,
\end{equation}
where the quasi-Lagrangian velocity is given by 
\begin{equation}  
    \Delta U_{0,\eta}(r,dt) := \int_{\R} \left( \phi_\eta(r - y) -  
    \phi_\eta(-y) \right) \mW(dy, dt),
\end{equation}
and $\beta$ is a Brownian motion independent of the Brownian sheet $\mW$.
The small-scale regularization renders the field  
$\Delta U_{0,\eta}(x, dt)$ Lipschitz continuous, allowing one to  
apply classical results from stochastic analysis to guarantee 
the existence of a unique solution to Eq.\eqref{eq:ql_model2}.
The process $R(t)$ is a one-dimensional diffusion process taking values on $\R$,
with generator
\begin{equation}\label{eq:quad-gen}
	\mL_K = \mA_{\eta,\kappa}(r)\frac{\mathrm{d^2}}{\mathrm{d}r^2}, 
\end{equation}
where the regularized diffusion coefficient is given by 
\begin{equation}\label{eq:coefficient kraichnan}
	\mA_{\eta,\kappa}(r)=\frac{1}{2}\int_\R dz\left[\phi_\eta(r-z) - \phi_\eta(-z)\right]^2
    + \kappa.
\end{equation}
The (symmetric) coefficient  $\mA_{\eta,\kappa}(r)$ 
represents the second-order structure function 
$\E\left(U_{0,\eta}(x+r)-U_{0,\eta}(x)\right)^2$
of the Kraichnan velocity field \eqref{eq:u0}, 
under the $(\eta,\kappa)$-regularization.

\subsubsection{Characterization as a time-changed Brownian motion}

In order to investigate the phase transitions in the 
one-dimensional Kraichnan model, it is useful to introduce the 
speed measure of the regularized separation process $R$.
The speed measure, denoted by ${\frak m}_{\eta,\kappa}$,
quantifies how quickly the process moves through different regions of 
the state space compared with a Brownian motion. Large values of the speed measure 
correspond to regions where particles slowly separate, 
while small values indicate regions where separation  
occurs more rapidly.
It recovers the escape time from $r$ out of an interval  $(r_1,r_2)\ni r$  as
\be
	\label{eq:escape}
	\mathbb E_r T_{r_1} \wedge T_{r_2}  = \int_\R G(r,y) {\frak m}_{\eta,\kappa}(dy), 
\ee
where $T_\rho := \inf\{ t > 0 : R(t) = \rho \}$ denotes the  
first hitting time of separation $\rho$, and $a \wedge b$ stands for  
the minimum of $a$ and $b$. The function $G$ is the Green function  
for the one-dimensional Laplacian on $(r_1,r_2)$ with Dirichlet boundary  
conditions,
conventionally defined as the solution to
\be
	- \frac{1}{2} \partial_{rr} G =\delta(r-z) \quad \text{ if } r,z \in (r_1,r_2) \quad \text{and} \quad G(r,z)= 0 \quad \text{otherwise},
\ee
and with explicit expression 
\begin{equation}
	G(r,z) = \frac{2(r\wedge z-r_1)(r_2-r\vee z)}{(r_2-r_1)},
\end{equation}
where it is not vanishing.
In particular, the Brownian escape times are recovered by  applying the transformation  ${\frak m}_{\eta,\kappa}(dz)  \mapsto dz$  
in Eq.~\eqref{eq:escape}.\\

For the regularized separation process~\eqref{eq:ql_model2}, the speed measure is given by  
(cf.~\cite{revuz2013continuous}, Exercise~VII.3.20)
\begin{equation}\label{eq:speed_regularized}
	{\frak m}_{\eta,\kappa}(dr) = \dfrac{dr}{2 \mA_{\eta,\kappa}(r)},
\end{equation}
and provides a complete description of the statistics  
of the inter-particle separation in the one-dimensional Kraichnan model.
In fact, the infinitesimal generator $\mL_K$ can be written in terms  
of ${\frak m}_{\eta,\kappa}$ alone as
%\footnote{Equation \ref{eq:generator_speed} is a formal expression that is to be interpreted in the weak sense, that is, integrated against suitable test functions \cite{breiman1992probability,revuz2013continuous}.}
\be
    \label{eq:generator_speed}  
    \mL_K = \dfrac{1}{2} \, \frac{\mathrm{d}}{\mathrm{d} \frak{m}_{\eta,\kappa}}  
    \frac{\mathrm{d}}{\mathrm{d} r}.
\ee
At a fundamental level, the relevance of the speed measure stems from  
its role in representing the separation process as a time-changed  
Brownian motion. 
Indeed, 
an application of Itô's formula, 
as presented in~\cite[pp.~58--59]{bakry2013analysis},
leads to the following equality in law 
\begin{equation}\label{eq:R=BM}
    R(t) \stackrel{\text{law}}{=} B(\tau_{\eta,\kappa}(t)),  
\end{equation}
where $B$ is a Brownian motion and 
\be
	\label{eq:dds}
	\tau_{\eta,\kappa} = \inf \left\lbrace s\ge 0, \mC_{\eta,\kappa}(s)>t \right\rbrace,\quad \mC_{\eta,\kappa}(t)=\int_0^t \dfrac{ds}{2\mA_{\eta,\kappa}(B_s)}.
\ee
In this representation, $\tau_{\eta,\kappa}$ is the  
quadratic variation of $R$,  
prescribing the growth $R(t)^2 \sim \tau_{\eta,\kappa}(t)$.  
It can be heuristically interpreted  
as the natural (proper) timescale of the process.  
On the other hand, $C_{\eta,\kappa}$ can be viewed as a clock process,  
measuring how long the separation process takes to traverse each portion  
of the state space. 
Both quantities are non-decreasing functions of time,  
and expression \eqref{eq:dds} means that  
$\tau_{\eta,\kappa}$ is the generalized right-continuous inverse  
of $C_{\eta,\kappa}$.
Observe that, due to the presence of regularization,  
the coefficient $\mA_{\eta,\kappa}(r)$ is strictly positive and  
uniformly bounded below in $r$. This implies that  
$\C_{\eta,\kappa}(t) < \infty$ for each $t$. In other words,  
the separation process does not get stuck in any region of the state space, 
and in particular,  
particles do not coalesce upon collision.  
We stress that representation \eqref{eq:R=BM}
is non-trivial: it
allows one to construct realizations of the regularized separation process 
using the Brownian path $B$ as the only input. 
In doing so, the clock process is expressed entirely as a functional  
of Brownian motion, and the resulting separation process is measurable  
with respect to $B$.

The role of the speed measure becomes explicit upon introducing 
the Brownian local time $\ell(r,t)$,  which we recall is  associated to the amount of time the Brownian spends at position $r$. 
By applying the occupation times formula \cite[Corollary VI.1.6]{revuz2013continuous}, 
the clock process can be represented as  (see, \emph{e.g.} \cite[Chapter IV]{borodin2015handbook})
\be
	\label{eq:tauc-m}
	\mC_{\eta,\kappa}(t) = \int_\R \ell(y,t)\frak m_{\eta,\kappa}(dy),\quad \ell(y,t) = \lim_{\epsilon \downarrow 0} \dfrac{1}{2\epsilon}\int_0^t ds\,\mathds{1}_{B_s \in (y-\epsilon,y+\epsilon)}.
\ee
Thus, in addition to being defined by the SDE \eqref{eq:ql_model2},  
the regularized inter-particle separation can also be realized  
as  time-changed Brownian motions, parametrized  through the speed measures \eqref{eq:speed_regularized} 
and relations \eqref{eq:R=BM}--\eqref{eq:tauc-m}. This  fundamental viewpoint  will be used extensively in this work.

\subsection{The unregularized separation process}\label{ssec:unregularizedKraichnan}
\subsubsection{Definition}
Formulation \eqref{eq:R=BM} 
provides a convenient framework for defining the 
unregularized separation process.
By setting $\eta = \kappa = 0$ in expression \eqref{eq:coefficient kraichnan},
one obtains the unregularized diffusion coefficient
$\mA \equiv \mA_{0,0}$. Its small-scale behavior is characterized by

\begin{lemma}\label{thm:A bound}
    Let $\xi\in(0,2]$. There exist constants 
    $C_+,C_->0$ such that, for all $r \in[-L,L]$,\\
\be   
	 C_- \left(\frac{|r|}{L}\right)^{\xi} \leq  \mA(r)\leq C_+ \left(\frac{|r|}{L}\right)^{\xi}.
   \ee  
\end{lemma}
In short, $\mA \approx (|r|/L)^\xi$ on $[-L,L]$.  This property comes from the direct calculation  \cite{robert2008hydrodynamic}
\be
	\label{eq:cd}
	 \mA(r) = c_d\left(\frac{|r|}{L} \right)\left(\dfrac{|r|}{L}\right)^\xi,\quad c_d(\rho) :=  \dfrac{1}{2} \int_\R dz\,\left|  \psi(\rho+\rho z)\dfrac{1+z}{|1+z|^{3/2-\xi/2}}-\psi(\rho z)\dfrac{z}{|z|^{3/2-\xi/2}}\right|^2,
\ee
together with the observation that  $c_d(\rho)$ is bounded on (0,1)---see Figure \ref{fig:L1} and proof in in \cref{sec:proofs}. 

Associated to $\mA$ are the unregularized speed measure
\begin{equation}\label{eq:speed3} 
	\frak{m}(dr)\equiv \frak{m}_{0,0}(dr)=\frac{dr}{2\mA}, 
\end{equation}
and the unregularized clock process 
\be \label{eq:unreg_clock}
	\mC(t)\equiv\mC_{0,0}(t) = \int_\R \ell(y,t)\frak m(dy).
\ee
The unregularized separation process is then \emph{defined}  
as the time-changed Brownian motion associated with the 
clock process \eqref{eq:unreg_clock}, through the unregularized counterpart of 
relation \eqref{eq:R=BM}. This defines a Markov process with continuous sample paths
for every typical realization of $B$, which formally solves the SDE \eqref{eq:ql_model2} with $\eta=\kappa=0$.

The clock process $\C$ is almost surely finite for times shorter  than the first hitting time of the origin by $B(t)$.
However, since $\mA(0)=0$,   it may attain infinite values when $B(t)$ hits the origin. This corresponds to the situation where particles collide with each other  and coalesce. This question will be addressed below, and goes along with the  identification of the various phases for  the unregularized Lagrangian flow in the one-dimensional Kraichnan model.

\begin{figure}
	\centering
	\includegraphics[width=0.42\textwidth]{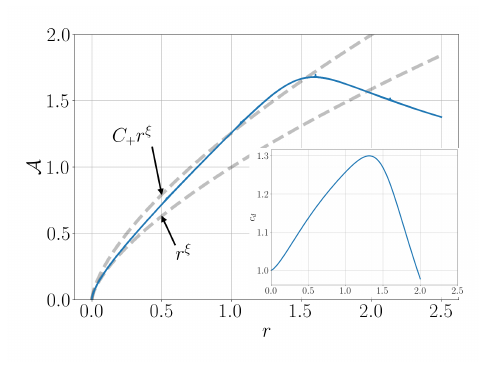}
	\includegraphics[width=0.42\textwidth]{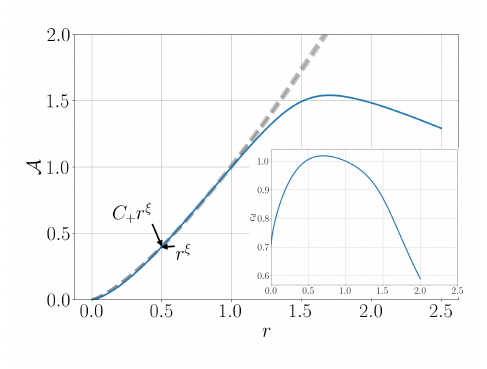}\\
	\includegraphics[width=0.42\textwidth]{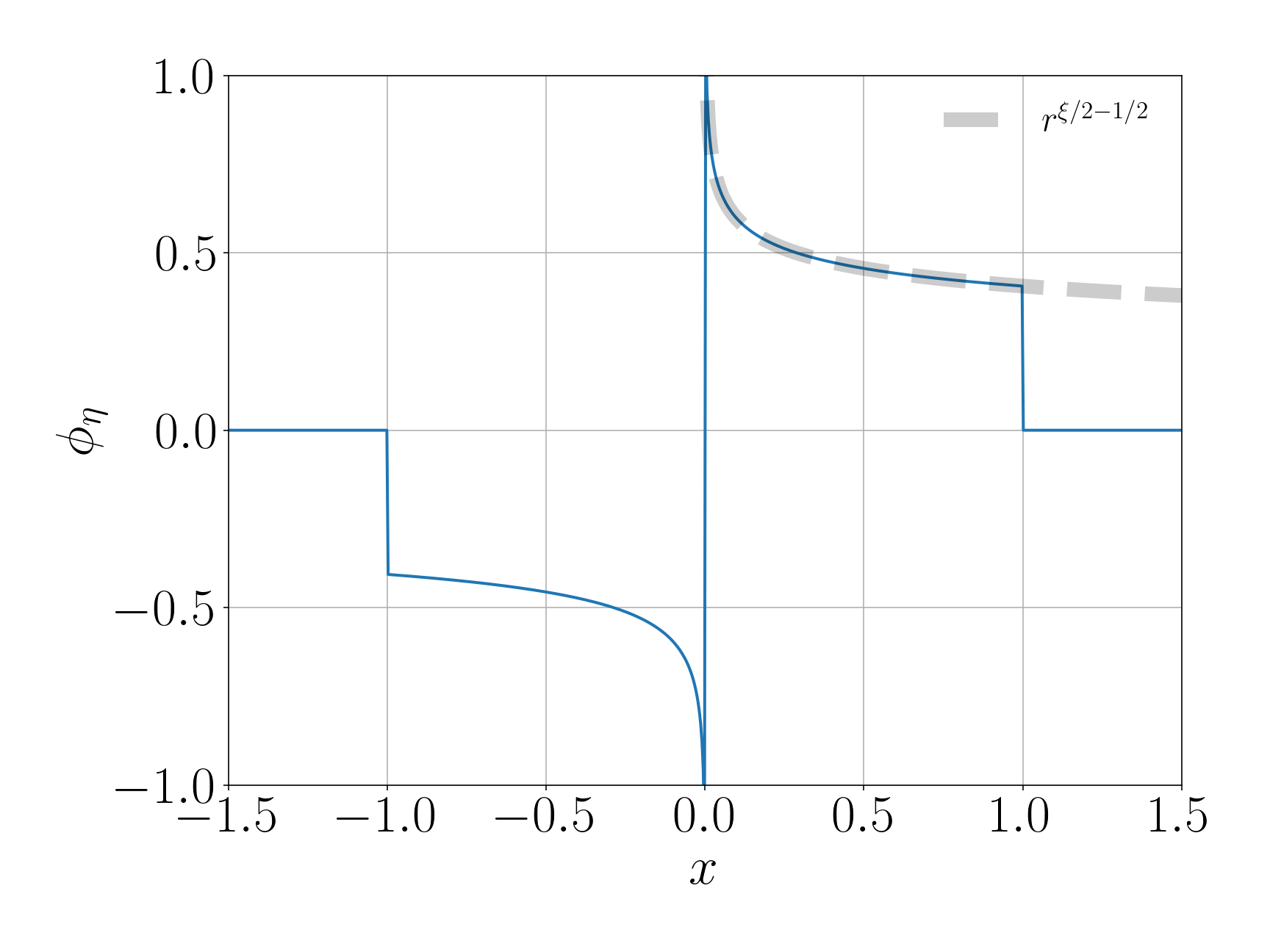}
	\includegraphics[width=0.42\textwidth]{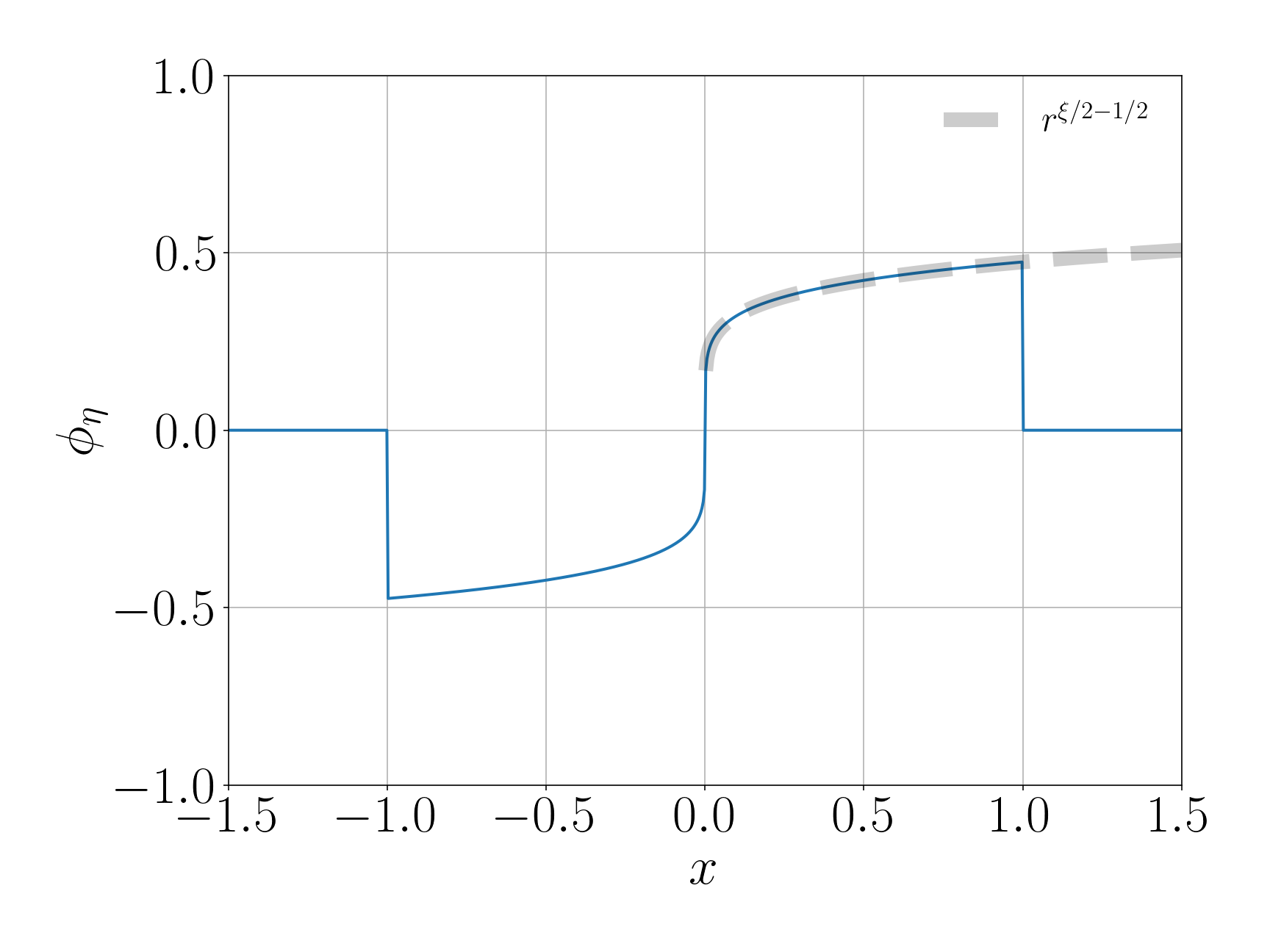}
	\caption{Top panels show the scaling coefficient $\mA$, with numerical estimates of the prefactors $c_d(r)$ given by Eq.~\eqref{eq:cd} in insets for $\xi =2/3$ (left) and $\xi=4/3$ (right). The bottom panels show the corresponding RV kernels. Numerics use $\eta=10^{-10}$ and   $\psi(x) = 1/2-\tanh(4x-1.5)/2$ }
	\label{fig:L1}
\end{figure}

\subsubsection{Boundary behaviors}
The different regimes of the unregularized Lagrangian flow
can be derived by applying  Feller's theory of one-dimensional diffusions 
and their boundary behavior  \cite{feller1954diffusion,breiman1992probability,ito2012diffusion} 
to the unregularized separation process.
In principle, the process $R$ can take both positive and negative 
values. 
However, since our focus is on the explosive separation or collapse of trajectories,
we may as well restrict $R$ to the positive half-line: this is equivalent to studying the norm of the separation process. 
The norm process is obtained by replacing the Brownian motion featured in the time-changed characterization  \eqref{eq:R=BM} by a Brownian motion reflected at zero.
Upon restricting it to the positive half-line, the pair-separation  process $R$ is a one-dimensional diffusion with state space  
contained in the interval $[0,\infty)$, and is fully characterized by its speed measure  
\begin{equation}\label{eq: full speed2}  
    \frak M(dr) = \mathds{1}_{(0,\infty)}(r)\, \frak{m}(dr)  
    + \frak M_0 \delta(dr),  
\end{equation}  
where $\frak{m}(dr)$ is given by \eqref{eq:speed3}, and  
$\delta(dr)$ denotes the Dirac measure at zero.  
While $\frak{m}(dr)$ governs the fluctuations of $R$ within the  
interval $(0,\infty)$, 
the atomic part $\frak M_0 \in [0,\infty]$ is related to the post-collisional 
behavior of fluid trajectories. 
Whether $\frak M_0$ plays a role, however, depends on 
the nature of the origin as a boundary point, within the 
framework introduced by Feller \cite{feller1954diffusion,breiman1992probability},
which we now study. \\

The phases of the Lagrangian flow are determined by the boundary type at $r = 0$.
Following the presentation of Breiman~\cite[Chapter~16]{breiman1992probability}, 
the origin is characterized according to its accessibility properties 
to and from non-vanishing separations $r \in (0, \infty)$,
expressed through the hitting times  $T_r:=\inf\{t>0: R(t)=r\}$ and conditional probabilities 
$\P_{r}\left(\cdot\right):=\P\left(\;\;\cdot \;\;\middle| \,R(0) = r\right)$.

The first criterion determines whether the point $0$ can be reached in finite time,
and is formalized in next definition.
\begin{definition}
	 The point $0$ is said to be an \textit{inaccessible (accessible) boundary}  if,  for all $r>0$,
	\begin{equation}
		\P_{r}\left(T_0< +\infty\right)\underset{(>)}{=}0. 
	\end{equation}
\end{definition}
If  accessible, $0$ can be reached in finite time with positive probability, thus belonging to the state space. %Accessible boundary points are sometimes called \emph{closed}.
In our context of particle separation, this situation describes finite-time collision between any two fluid particles.
By contrast, if $0$ is inaccessible, it cannot be reached in finite time and does not belong to the state space. 

The second criterion determines whether non-colliding trajectories can branch out in finite time,  
and leads to the following classification.
\begin{definition}
	Let $0$ be an inaccessible boundary. We call it \emph{natural} if, for all  $r, t > 0$,  
	\begin{equation}
		\lim_{r_0 \to 0^+} \P_{r_0}(T_r < t) = 0,
	\end{equation}
	and we call it \emph{entrance} if, for all $r > 0$, there exists $t > 0$ such that  
	\begin{equation}
		\lim_{r_0 \to 0^+} \P_{r_0}(T_r < t) > 0.
	\end{equation} 
\end{definition}
The boundary point $0$ being natural means that particles  
starting from distinct initial positions never collide, nor do fluid particles split trajectories. This regime intuitively connects with existence and uniqueness of solutions  to the flow equations. In contrast, if $0$ is an entrance boundary, particles never  collapse, but fluid trajectories may split, branching into distinct   paths in finite time --- thus giving rise to an effective stochastic  behavior.

Similarly, when  accessible, $0$  is called  \emph{regular} if  particles can escape and  \emph{exit}  otherwise.
In this case, these properties are determined by the speed measure \eqref{eq: full speed2}  
in neighborhoods of the origin, according to the next definition.
\begin{definition}
	Let $0$ be an accessible boundary. We call it \emph{exit (regular)}  if, for  
	all $\delta > 0$,  
	\begin{equation}\label{eq:exit}
		\frak{m}((0, \delta)) \underset{(<)}{=} +\infty.  
	\end{equation}	
\end{definition}
An exit boundary point is a point of no return: it can be reached in  
finite time, but the process cannot start from it.  That ${\frak m}$ takes arbitrarily large values in arbitrarily small  neighborhoods of $0$ means that the process spends arbitrarily large  
amounts of time near the origin --- regardless of the value assigned to  
the atomic part $\frak M_{0}$.  
In the context of Kraichnan flows, this situation corresponds to  different particle trajectories sticking together upon collision and  
coalescing into a single trajectory. 
In contrast, $\frak{m}$ is finite in neighborhoods of regular boundary points. 
In this case, whether the process spends an infinite amount of  time at the boundary is determined by the atomic component  $\frak{M}_0$.
If $0$ is a regular boundary, 
the quantity $\frak{M}_0$ 
decides the fate of fluid particles upon collision, 
ranging from coalescence ($\frak M_0=\infty$) to branching ($\frak M_0=0$). \\

The following diagram schematically summarizes the possible boundary  
types, with the arrow symbols denoting accessibility properties from  
and to $0$:  
\be
	\rotatebox{90}{\hspace{-0.75cm} \text{accessible} }
	\left \lbrace\hspace{1cm}
	\begin{split}
	& \text{regular:}\; \left\lbrace0\right\rbrace \;{ \huge \substack{\rightarrow\\ \leftarrow } }\;(0,\infty) \quad \quad\quad
	\text{entrance:}\; \left\lbrace0\right\rbrace \;{ \huge \substack{\rightarrow\\ \nleftarrow } }\;(0,\infty) \\
	& \text{exit:}\; \left\lbrace0\right\rbrace \;{ \huge \substack{\nrightarrow\\ \leftarrow } }\;(0,\infty) \quad \quad\quad\quad\,\,
	\text{natural:}\; \left\lbrace0\right\rbrace \;{ \huge \substack{\nrightarrow\\ \nleftarrow } }\;(0,\infty). \\
	\end{split}
	\hspace{1cm}
	\right\rbrace
	\rotatebox{270}{\hspace{-1cm} \text{inaccessible} }
\ee

%\begin{equation*}
%	\begin{split}
%\forall r >0\;\quad \quad
%	& {\rightarrow}:  \quad\exists T<\infty\, \lim_{\eta\downarrow 0} \P_{\eta}\left(\tau_r < T\right) >0\\
%	& {\nrightarrow}:\quad  \forall T<\infty\, \lim_{\eta\downarrow 0} \P_{\eta}\left(\tau_r < T\right) =0\\
%	& {\leftarrow}:  \quad\exists T<\infty\, \lim_{\eta\downarrow 0} \P_{r}\left(\tau_\eta < T\right) >0\\
%	& {\nleftarrow}:\quad  \forall T<\infty\, \lim_{\eta\downarrow 0} \P_{r}\left(\tau_\eta < T\right) =0\\
%	\end{split}
%\end{equation*}

\subsubsection{Phases of the  Lagrangian flow}
\label{ssec:phasesmonofractal}
The different regimes of the Lagrangian flow are determined
by two useful lemmas, which provide convenient  
computational criteria for determining the boundary type at the origin. 
\begin{lemma}\label{thm:closed_point}
    The point $0$ is an inaccessible (accessible) 
	boundary if and only if, $\forall \delta>0$,
    $$
    \int_{0}^{\delta}r \;{\frak m}(dr) \underset{(<)}{=}+\infty.
    $$
\end{lemma}

\begin{lemma}\label{thm:natural}
If $0$ is inaccessible, then it is natural.
\end{lemma}
\begin{remark}
	\cref{thm:natural} is specific to the one-dimensional Kraichnan model,  
	where the unregularized separation process has no drift component,
	and is an immediate consequence of \cite[Prop.~16.45]{breiman1992probability}.  
	In the multidimensional version of the Kraichnan model, however,  
	the origin being inaccessible does not imply that it is natural.  
\end{remark}

We can now classify the possible behaviors of the boundary point $r=0$ for the unregularized separation process $R$.
This classification precisely reflects the stochastic/deterministic and colliding/non-colliding dichotomies 
at the level of the Lagrangian flow, expressed previously in terms of the transition kernels in Eqs.~\eqref{eq:dicho2} and \eqref{eq:dichotomy}.
%\begin{theorem}\label{thm:phase}
%    Let $\xi\in (0,2]$. The boundary point $0$ is:\,
%\emph{(I)} regular if $\xi<1$, \quad 
%\emph{(II)}  exit  if $1 \leq \xi<2$,\quad and \quad
%\emph{(III)} natural  if $\xi=2$.
%\end{theorem}
\begin{theorem}\label{thm:phase}
	Let $\xi \in (0,2]$. The boundary point $0$ is
		\emph{(I)} \emph{regular}, if $\xi < 1$;\quad 
		\emph{(II)} \emph{exit}, if $1 \leq \xi < 2$;\quad
		\emph{(III)} \emph{natural}, if $\xi = 2$.

%	\begin{itemize}
%		\item[\emph{(I)}] \emph{regular}, if $\xi < 1$;
%		\item[\emph{(II)}] \emph{exit}, if $1 \leq \xi < 2$;
%		\item[\emph{(III)}] \emph{natural}, if $\xi = 2$.
%	\end{itemize}
\end{theorem}
\begin{proof}
The proof relies on the fundamental property $\mA(r) \approx r^\xi$ on $(0,L)$,
which implies ${\frak m}(dr) \approx r^{-\xi} dr$ on any interval $(0,\delta) \subset (0,L)$.
We thus have $\int_{0}^{\delta}r \;{\frak m}(dr)< +\infty$ if and only if $\xi<2$. From \cref{thm:closed_point} we conclude that $0$ is accessible if and only if $\xi<2$. We then use \cref{thm:natural} to conclude that
the point $0$ is natural for $\xi=2$, proving (III). For $ \xi<2$, we estimate the speed measure as  ${\frak m}((0,\delta)) \approx \int_0^\delta r^{-\xi} dr$. From  \eqref{eq:exit}, we deduce that the point $0$ is regular if $\xi<1$ and exit if $1\leq\xi<2$,
leading to (I) and (II).
\end{proof}
The physical interpretation of \cref{thm:phase} goes as follows.
For sufficiently rough flows ($\xi<1$), trajectories collide in finite-time 
almost surely. Their subsequent behavior is determined by the choice of ${\frak M_0} \in [0,\infty]$. 
Setting ${\frak M_0}=0$
corresponds to trajectories reflecting  
instantaneously from the origin after collision, leading to random  
evolution of particles in quenched space-time realizations of the  
velocity field, even in the absence of thermal noise.  
This marks the emergence of spontaneous stochasticity.
On the other hand, assigning ${\frak M_0} = \infty$ corresponds to  
particles sticking together upon collision and coalescing into a  
single trajectory.
The case $0<{\frak M_0}<\infty$ describes soft collisions and interpolates 
between the other two cases.
For smoother flows ($1 \leq \xi < 2$), trajectories collide and  
coalesce almost surely in finite time, regardless of the choice of  
${\frak M_0}$.  
For spatially smooth flows ($\xi = 2$), the point $0$ is a natural  
boundary. In this phase, trajectories neither collide nor branch.  
This reflects the well-posedness of the initial value problem for  
Eq.~\eqref{eq:ql_transport}, where two particles starting from the  
same position ultimately follow the same trajectory in a prescribed  
realization of the velocity field.  
As such, the threshold $\xi = 1$ marks the dichotomy \eqref{eq:dicho2}  
between deterministic and spontaneously stochastic behavior, while  
$\xi = 2$ marks the dichotomy \eqref{eq:dichotomy} between colliding and  
non-colliding behavior.

In short, and in physical terms, the boundary behaviors described 
by \cref{thm:phase} distinguish between three 
different phases for the Lagrangian flow, which we refer to as 
(I) branching (colliding/conditionally stochastic),  (II) coalescing (colliding/deterministic) and (III)  smooth  (non-colliding/determistic). 
The coalescing and smooth phases are universal with respect to regularization schemes,
whereas the branching phase is not (see \cref{ssec:vanishingreg} below).

\subsubsection{The Bessel shortcut to the phase diagram}
\label{ssec:Besselshortcut}
A shortcut  to  Theorem \ref{thm:phase} \cite{gawedzki2000phase,falkovich2001particles,gawedzki2002soluble} is to realize that, 
up to the first  collision time $T_0$, 
the dynamics behaves essentially as the multiplicative process 
\be 
	\label{eq:MBM}
	dR = R^{\xi/2}dW,\quad R(t=0)=r,
\ee 
for initial relative separation $r>0$.
It\^o's formula then tells us that, for $t<T_0$ , the  process $X_t=\dfrac{1}{1-\xi/2} R_t^{1-\xi/2}$ is the  Bessel process of parameter $a$, defined through
\be
	\label{eq:bessel}
	dX =\dfrac{a}{X} dt +  dW  ,\quad a = \dfrac{\xi}{2\xi-4},
\ee
representing the norm of a Brownian motion in effective dimension $d_e=2a+1 = \dfrac{2-2\xi}{2-\xi}$.
The properties of Bessel processes have now been extensively studied --- see, \emph{e.g.}, \cite{lawler2018notes}.
In particular, those processes do not reach the origin for $a\ge 1/2$, are reflected about the origin for $-1/2<a<1/2$, and are absorbed by the origin for $a<-1/2$. As a function of $\xi$, the parameter $a$ decreases from $0$ to  $-\infty$ as $\xi$ goes from $1$ to $2$. 
The stochastic and colliding dichotomies at $\xi=1$ and $\xi=2$ precisely correspond to the transitions at $a=-1/2$ and $a=-\infty$, respectively --- see the left panel of Fig.~\ref{fig:comments}.

As sidenote, the mapping of particle separation to Bessel processes can 
be made also a physical  one,   
if one considers a \emph{random potential version} of  
the field \eqref{eq:u0}. The model replaces the space-time white noise,  
within the quasi-Lagrangian field \eqref{eq:ql_model2}, 
by a Brownian path $\widetilde {\mathcal W}(dy,dt) = \delta(y-R) \,dW dy$ 
attached to the end-point separation 
--- this changes $\mathcal W$ to $\widetilde{ \mathcal W}$ in \eqref{eq:u0}. 
Upon using the random potential $\widetilde \phi(r) = |r|^{\xi/2} \wedge L$, 
the (norm of the) separation process is then precisely a Bessel process   
up to the first exit time out of the interval $[0,L]$.

\begin{figure}[h!]
    \centering
    \includegraphics[width=0.49\textwidth]{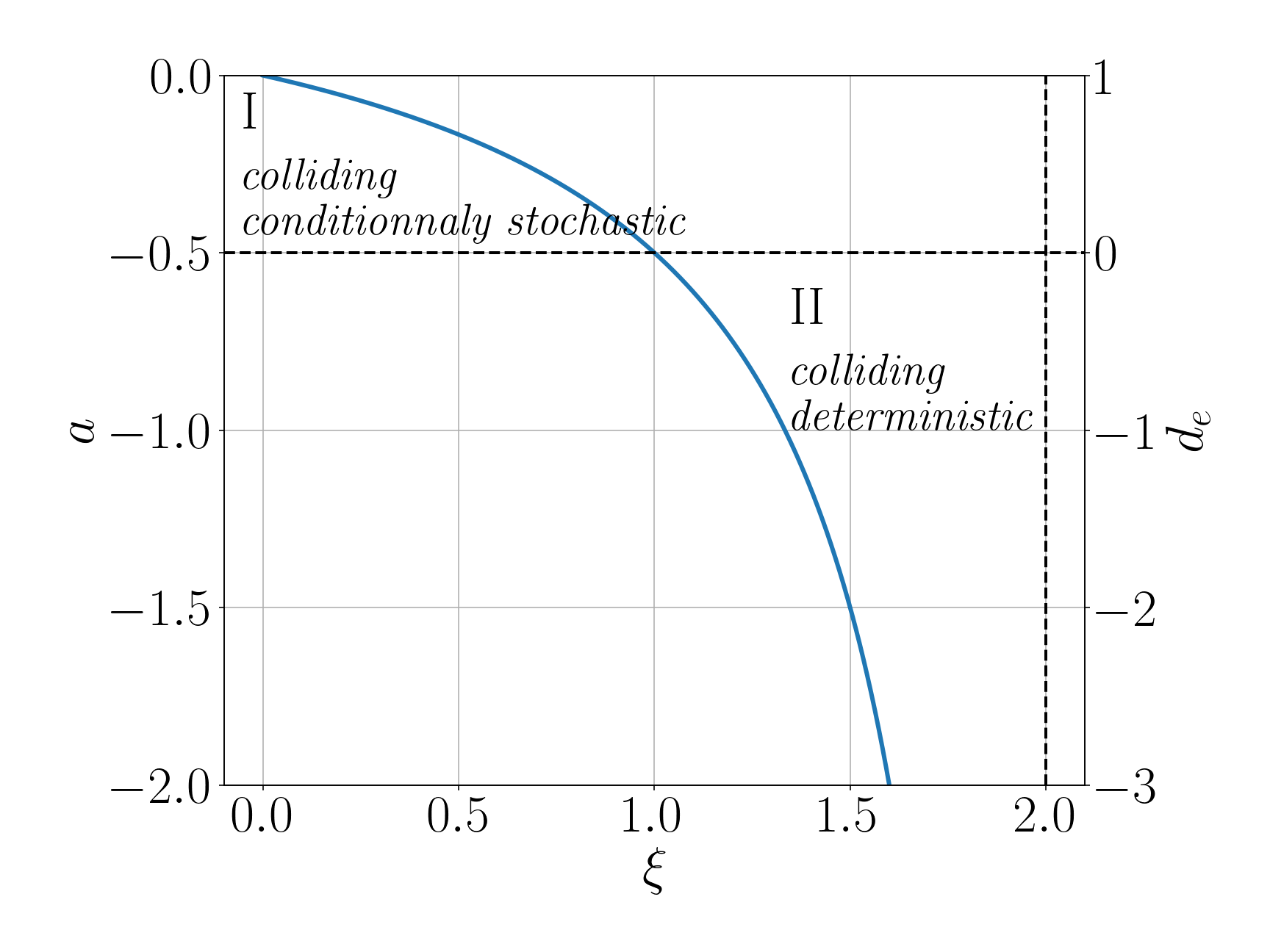}
    \includegraphics[width=0.49\textwidth]{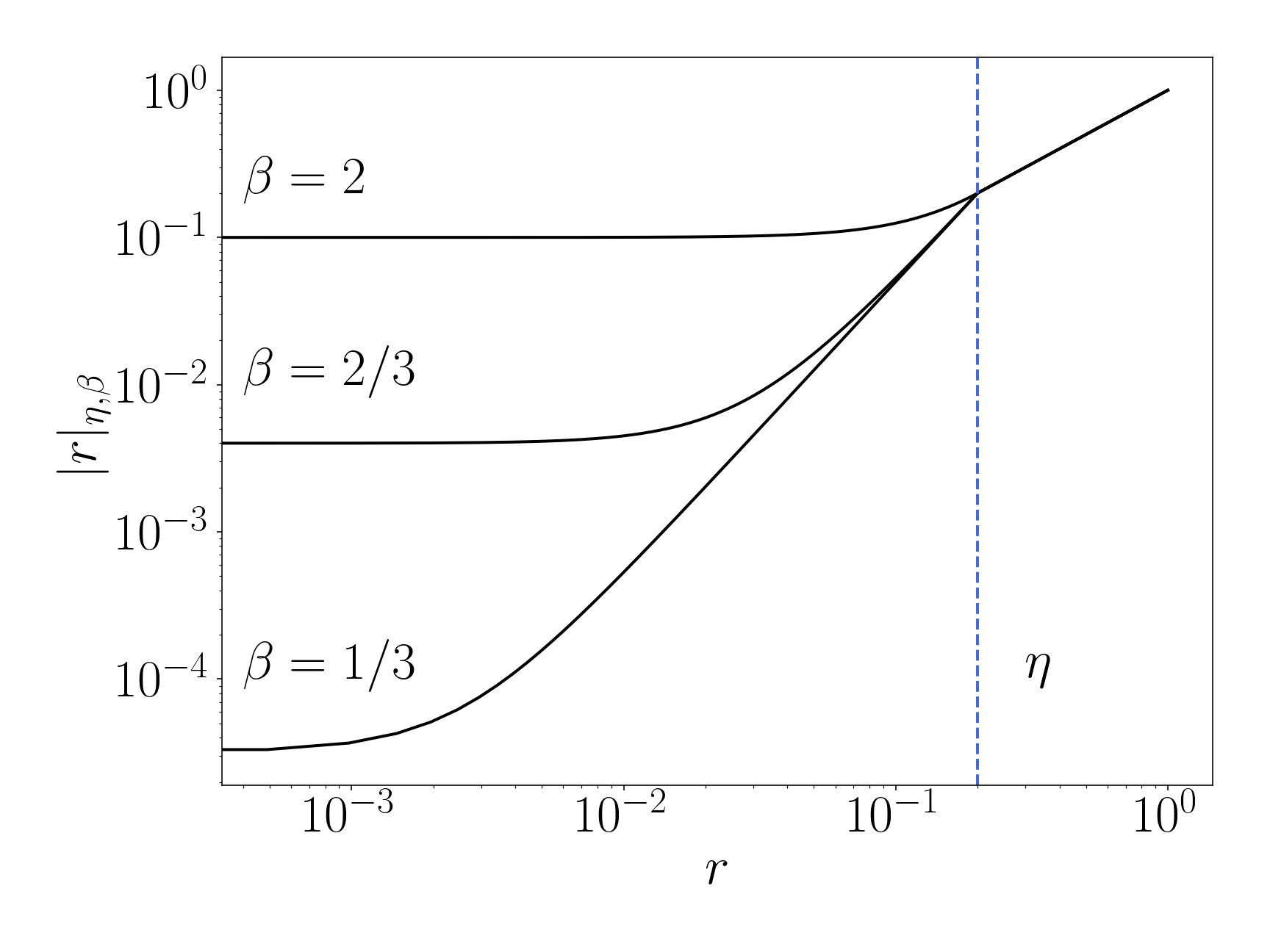}
    \caption{Left: Effective parameters $a$ and $d_e$  of the  dual Bessel process \eqref{eq:bessel}. Right: Regularized norms $|\cdot|_{\eta,\beta}$ for  $\eta =0.2$  and various $\beta>0$}
    \label{fig:comments}
\end{figure}

\subsection{The limit of vanishing regularization}
\label{ssec:vanishingreg}
The freedom in the choice $\frak M_0$ at the level of the unregularized 
processes for the phase $0<\xi <1$ connects to the boundary behaviors 
in the boundary layer $(0,\eta)$ for the  pair separation. 
It reflects the non-universality of the vanishing $\eta,\kappa$ limit, 
with the limiting dynamics sensitive to the small-scale details. 
In other words, different regularizations lead to different limiting processes.
To illustrate this feature, let us decompose the speed measure of the 
unregularized process as 
\be
  \frak m_{\eta,\kappa}(dr)  = {\frak m_+}(dr)+{\frak m}_{\eta,\kappa}^{(0)}(dr),\quad \text{with} \quad {\frak m_{\eta,\kappa}^+}(dr)= \mathds{1}_{\R\backslash (\eta,\eta)}(r)\,{\frak m}(dr),\quad {\frak m^{(0)}_{\eta,\kappa}}(dr)= \mathds{1}_{(\eta,\eta)}(r)\,{\frak m}(dr),
\ee
and assume that on $(0,L)$ we have the relation
\be 
	{\mathcal A}_{\eta,\kappa}(r)-2\kappa \simeq |r|_{\eta,\beta}^\xi
\ee
with the mollifications $r_{\eta,\beta}$  defined  through the $\beta$-regularizations as
\begin{equation}\label{eq:regularized_norms}
    |r|_{\eta,\beta}=
    \begin{cases}
        \frac{|r|^2}{C \eta} + \frac{\eta^{2/\beta}}{2}\quad &\text{if}\quad r\leq \eta, \\
        |r|\quad&\text{if}\quad r> \eta,
    \end{cases},
    \quad C=\frac{2}{2-\eta^{2/\beta-1}},\quad \beta>0.
\end{equation}
The norms are continous and non-vanishing on $(0,L)$ ---see the right panel of Fig.~\ref{fig:comments}.

Then,  from the defining relation \eqref{eq:generator_speed} ${\frak m_{\eta,\kappa}}(dr)  = (2\mA_{\eta,\kappa})^{-1}\,dr$,
we deduce that  $\frak m_+ (dr) \to \frak m(dr)$ as $\eta,\kappa \to 0$. However the boundary component behaves as
\be
	 \frak m[(-\eta,\eta)]  = O\left(\dfrac{\eta }{\eta^{2\xi/\beta}+\kappa}\right) \quad \text{as $\eta\to 0$}.
\ee
Setting a limit with  $ \kappa = O(\eta^{2\xi/\beta})$, any of the following three  cases may occur, depending on the exponent $\beta$ prescribing the regularization.
\be
	 m_{\eta,\kappa}^{(0)}[(-\eta,\eta)] \to 
	\begin{cases}
		& 0	\hspace{1cm}	\beta>\xi\\
		 & O(1)	\hspace{1cm}	\beta=\xi\\
		 & \infty	\hspace{1cm}	\beta<\xi\\
	\end{cases},\quad \text{ as } \eta\to 0.
\ee

While the  absolutely continuous part of the  speed measure exhibit universality with respect to regularizations schemes,
the quantity $\frak M(\{0\}) = \lim_{\eta,\kappa\to 0 }m_{\eta,\kappa}^{(0)}[(-\eta,\eta)] $  depends on the small scale details. 
This feature is the underlying mechanism behind the lack of universal behavior 
in the $0<\xi<1$ phase of the unidimensional Kraichnan model, 
 extending to the intermediate compressibility regime in 
multidimensional Kraichnan flows \cite{gawedzki2000phase,gawkedzki2004sticky,e2000generalized}.

%%%%%%%%%%%%%%%%%%%%%%%%%%%%%%%%%%%%%%%%%
%%%%%%%%  3: MULTIFRACTAL KRAICHNAN (MFK)
%%%%%%%%%%%%%%%%%%%%%%%%%%%%%%%%%%%%%%%%%
\section{Multifractal Kraichnan flows} \label{sec:direct}
In this section, we  study the separation process driven by the one-dimensional velocity fields \eqref{eq:velocity}.  
This provides a generalization of the one-dimensional Kraichnan model discussed in \S\ref{sec:Kraichnan}, 
with  spatial modeling of intermittency in terms of GMC theory, 
whose essential properties and definitions are briefly recalled here, 
but otherwise exposed in more detail in the review \cite{rhodes2014gaussian}. 
Similarly to the previous section, 
we define the regularized and unregularized separation process in 
multifractal Kraichan flows in terms of time-changed Brownian motion, 
with the associated clock process featuring the GMC.

\subsection{Gaussian multiplicative chaos}\label{sec:gmc2}

We consider GMC measures defined on $\R$, 
formally written as 
\begin{equation}\label{eq:gmc2}
       \mu(dx)= e^{2\gamma \gf(x)}dx,
\end{equation}
where $\gf$ is a
centered Gaussian random field with correlation function
\begin{equation}\label{eq:log-kernel}
    \E\left[\gf(r)\gf(0)\right]=\log_{+}\frac{L}{|r|},\quad\quad  \log_+(r)=\max(\log(r),0),
\end{equation}
and $\gamma\in[0,\sqrt{2}/2)$ is a free parameter 
controlling the level of intermittency.
The difficulty in defining such a measure arises from the
logarithmic singularity of the correlation function \eqref{eq:log-kernel}
at short-distances, which causes the field 
$\gf$ to be distributional valued rather than 
pointwise defined. 
Furthermore, the already large fluctuations of the field $\gf$ 
are amplified when taking the exponential, 
rendering expression \eqref{eq:gmc2} non-rigorous.
Observe that, given the poor regularity of $\gf$, the correlation function
\eqref{eq:log-kernel} is to be interpreted in the weak sense, 
that is, integrated against suitable test functions.

The standard approach to address
this issue is to apply a regularization procedure.
One introduces a sequence of approximations $(\gf_\eta)_\eta$, 
idexed by a small parameter $\eta>0$,
which converges to $\gf$ as $\eta\to0$.
Additionally, one also introduces a regularizing sequence 
$Z_\eta \propto \eta^{-2\gamma^2}$,
responsible for compensating the wild fluctuations of $\gf_\eta$ 
during the limit procedure. As a result, one obtains 
the regularized GMC measure 
\begin{equation}\label{eq:gmc_regularized}
        \mu_\eta(dx)= Z_\eta^{-1} e^{2\gamma \gf_\eta(x)}dx.
\end{equation}
The convergence of $\mu_\eta$ towards a non-trivial measure
exhibiting multifractal properties was first established in 
Kahane's original work \cite{kahane1985chaos}, utilizing 
martingales approximations of the field $\gf$.
Since then, the convergence of GMC measures has been established  under different approximations schemes \cite{rhodes2014gaussian}, with the limiting measure  independent from  the specific scheme of regularization. This allows for a natural interpretation of the measure $\mu$, 
formally defined in \eqref{eq:gmc2}, as the weak limit of the regularized measures $\mu_\eta$ --- see Fig. \ref{fig:3} for an illustration.\\
\begin{figure}
 	\includegraphics[width=0.49\textwidth]{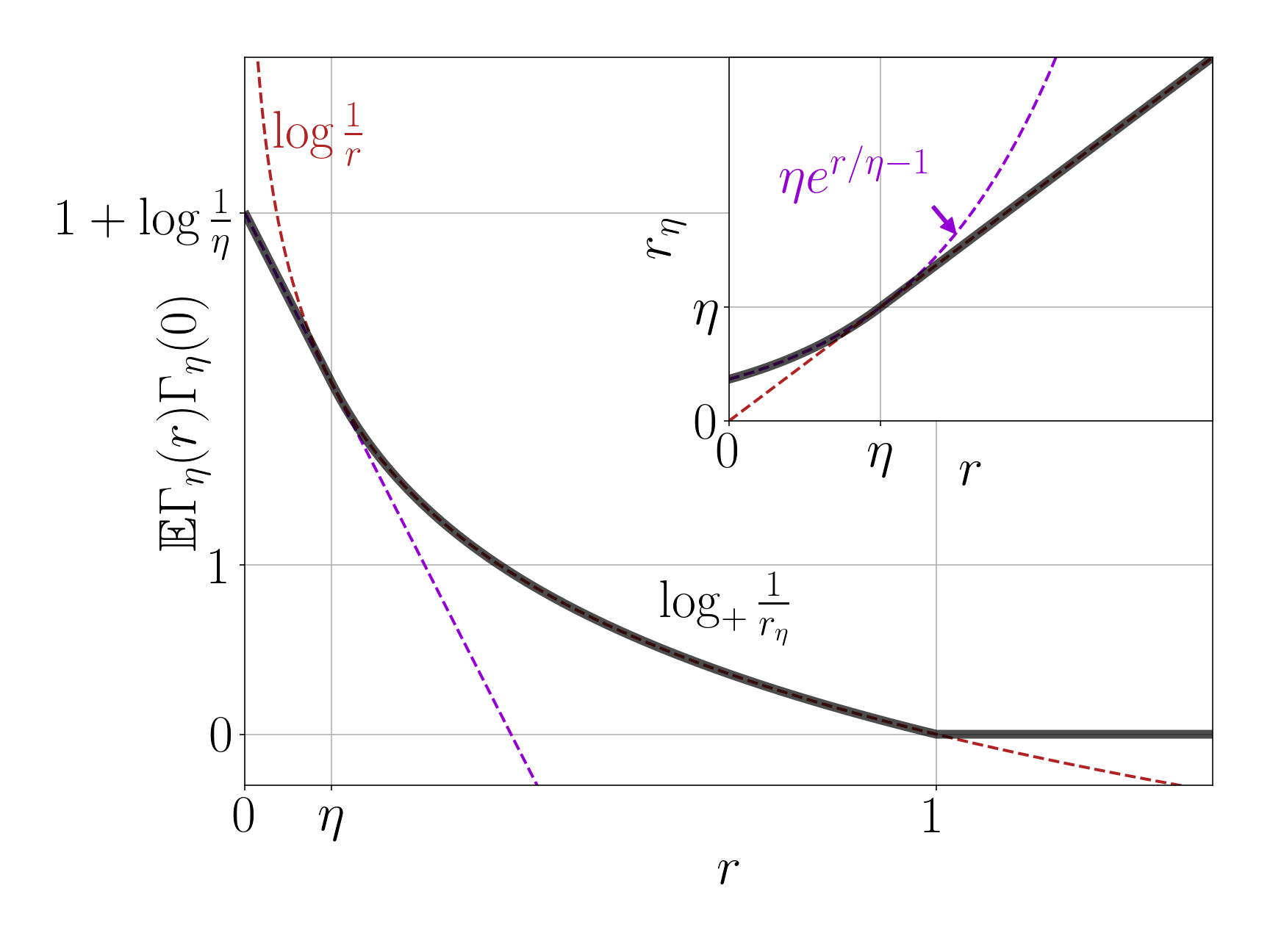}
 	\includegraphics[width=0.49\textwidth]{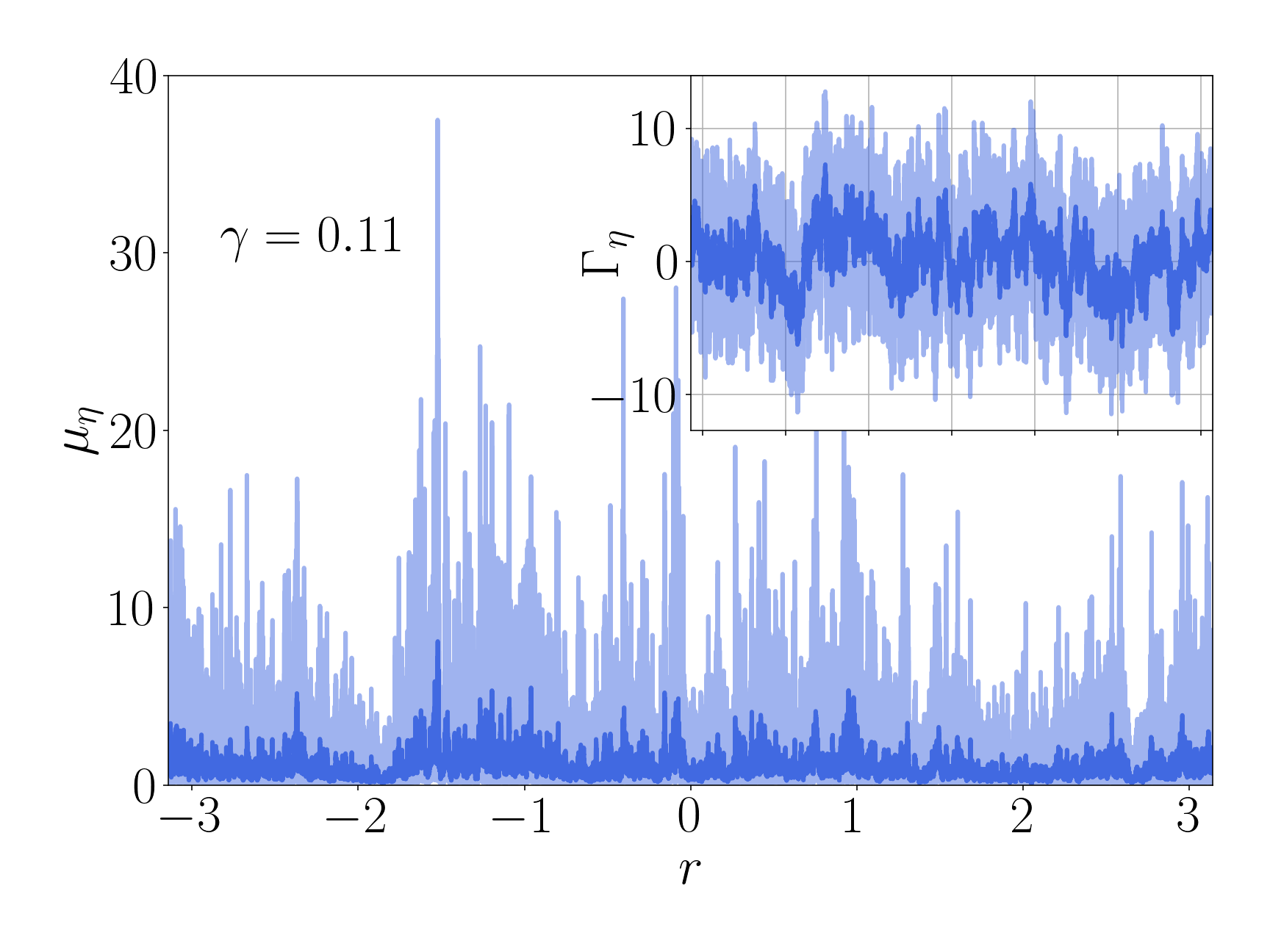}
	\caption{Left: $\eta$-regularization of the log kernel \eqref{eq:log-kernel} used in Kahane's martingale approximation. Right: Random realizations of the regularized GMC  $\mu_\eta  = (\eta/e)^{2\gamma^2} e^{2\gamma \Gamma_\eta} $ for $\eta=10^{-3}$ (deep blue) and $\epsilon = 10^{-6}$ (pale blue). Inset shows the regularized Gaussian field  $\Gamma_\eta$ with similar conventions.}
	\label{fig:3}
\end{figure}

One facet of the multifractal nature of GMC measures is the nonlinear behavior  
of the scaling exponents that characterize their mass distribution.
The GMC measure prescribed by expression \eqref{eq:gmc2} obeys the 
asymptotic relation \cite{rhodes2014gaussian} 
\begin{equation}
    \E\left(\mu([x,x+r])^p\right) \sim r^{p+\tau(p)}
    \quad \text{as}\quad r \to 0,
\end{equation}
%for $p\in \left(-\infty, \frac{1}{2\gamma^2}\right)$,
for $p\in \left(-\infty, \frac{1}{2\gamma^2}\right)$ and  with the scaling exponents
\begin{equation}\label{eq:power_law_spectrum_gmc}
%    \zeta_\mu(p) = \left(1 + \frac{\gamma^2}{2}\right)p 
%    - \frac{\gamma^2}{2}p^2,
    \tau(p) = 2\gamma^2p(1-p).
\end{equation}
The expression above prescribes a quadratic profile for the scaling exponents, with the parameter $\gamma$
accounting for deviations from the monofractal scaling $\E\left(\mu([x,x+r])^p\right) \sim r^p$.
In the context of  turbulence modeling, the local GMC averages $\mu_r(x):=r^{-1}\gmc\left([ x,x+r]\right)$ are  the central quantities connecting to Kolmogorov refined similarity theory of 3D homogeneous isotropic turbulence. They represent the power 2/3 of the  dissipation field coarse-grained over a distance $r$.
For the random field \eqref{eq:velocity}, the (rigorous) scaling \eqref{eq:uscaling} then  reflects the dimensional heuristics
\be
\label{eq:rss}
U(x+r)-U(x) \overset{law}\sim \mu_r(x)^{1/2} r^{\xi/2},
\ee
implying the bridging relation  $\zeta_U(p) = p\xi/2 + \tau(p/2)$ and tying 
together the scaling exponents of the dissipation and velocity fields. 
In particular, this relation predicts a negative  flatness exponent 
$\zeta_U(4) - 2  \zeta_U(2) = -16 \gamma^2$.  Choosing for example $ \gamma \simeq 0.07$  or $\gamma\simeq 0.11$ reproduces the flatness exponents  $\simeq - 0.1$ or $\simeq - 0.2$ respectively observed in  3D Navier-Stokes \cite{chevillard2015peinture} or surface quasi-geostrophic turbulence \cite{valade2025surface}. 

\cb

\subsection{The regularized multifractal separation process}

We consider the dynamics \eqref{eq:ql_transport}, driven by a  
regularized version of the multifractal velocity field  
\eqref{eq:velocity}, here denoted by $U_\eta$, obtained by replacing  
$\phi$ and $\gf$ with their regularized counterparts in  
\eqref{eq:velocity}. In this section, as in \cref{sec:Kraichnan}, we  
also employ the RV kernel \eqref{eq:rvkernel}. Recall that, in expression  
\eqref{eq:velocity}, the fields $\gf$ and $\mW$ are assumed to be  
independent.
The two-particle separation $R(t):=X_1(t)-X_0(t)$ then obeys the SDE
\begin{equation}\label{eq:ql_model-gamma}
    dR_t=  \Delta U_\eta(R_t,dt)+ \sqrt{2\kappa}\, d\beta, 
\end{equation}
where $\beta$ is a Brownian motion independent 
from $\gf$ and $\mW$, and  
the quasi-Lagrangian velocity is given by 
\begin{equation}\label{eq:ql_velocity-robert}
    \Delta U_\eta(r,dt):=   Z^{-1/2}\int_{\R} \left(\phi_\eta(r-y)-\phi_\eta(-y) \right)  e^{\gamma \Gamma_\eta(y)}\mW(dy,dt).
\end{equation}
Observe that, unlike in the (monofractal) Kraichnan setting, the  
interpretation of Eq.~\eqref{eq:ql_velocity-robert} as a  
quasi-Lagrangian field is crucial here to obtain a dynamical  
decoupling between the center of mass $(X_1 + X_0)/2$ and the relative  
separation, which in turn yields a closed dynamics for $R$.
This is due to the presence of the GMC term  
$\propto e^{\gamma \gf_\eta}$, which introduces a spatially  
non-homogeneous but time independent component in each realization  
of the velocity field \eqref{eq:ql_velocity-robert}.  
The quasi-Lagrangian formulation also emerges as the natural  
framework for addressing time-correlated extensions of the  
(monofractal) Kraichnan theory  
\cite{kupiainen2003nondeterministic,chaves2003lagrangian}.

The fact that the two-point function of $U_\eta$
is the same as that of the Kraichnan ensemble
raises the question of whether multifractality can 
lead to new phases of the Lagrangian flow or whether it 
has any impact on the two-particle motion at all.
Indeed, one of the takeaways of \cref{sec:Kraichnan} is that 
the different phases of the Kraichnan 
model are entirely prescribed by its correlation function and scaling exponent $\xi$. 
A naive guess based on these considerations would then 
speculate that the Lagrangian flow generated by $U^\eta$ should 
be the same as that generated by the Kraichnan ensemble, as a consequence of 
having $\zeta_U(2)=\xi$ (see \eqref{eq:zeta_U}).
This guess is however not correct, as the GMC induces non-trivial intermittent corrections.

Here, we  are mostly interested in the statistics of the 
inter-particle distance for quenched realizations of $\gf$.  
By \emph{quenched statistics}, we mean statistics conditioned on specific realization of  $\gf$.
By fixing $\gf$, we allow particles to interact with the GMC enviroment, 
and  intermittency will turn out to have a clear physical effect.  
The effect of intermittency will, however, also  be present when considering an \emph{annealed version} of the separation process, obtained by averaging out the environment. 
Although the physical meaning becomes less clear, this leads to a simpler mathematical framework, with also possible connections to 
homogenization theory.%, with its statistical behavior governed by a homogenized version 
%of the generator in Eq.\eqref{eq:generator_MF_Kraichnan} below.

To make the distincion between the various layers of randomness, we introduce the notation $\E^\gf$ for the expectation with respect to $\gf$,  and $\E^\mW$ for expectation with respect to both Brownian motion and the Brownian sheet. The corresponding probability measures are denoted by $\P^\gf$ and $\P^\mW$. Occasionally, we will also use $B$ to denote   Brownian trajectories. In such cases, the notation  for the probability measure and expectation  follows the same convention as for $\mW$.

From Eqs.\eqref{eq:ql_model-gamma} and \eqref{eq:ql_velocity-robert},  
we obtain the quenched quadratic variation of the regularized 
separation process
\begin{equation}\label{eq:quadratic_variation2-gam0}
        \biggl\langle R(t), R(t) \biggr\rangle =
        Z^{-1} \int_0^t\int_\R \left[\phi_\eta(R_s-z) - \phi_\eta(-z)\right]^2 e^{2\gamma\gf_\eta(z)}dzds
        + 2\kappa t.
\end{equation}
Applying Itô's formula for quenched realizations of $\gf_\eta$ 
leads to the generator of the regularized separation process
\begin{equation}\label{eq:generator_MF_Kraichnan}
    \mL_{\gamma} = \mA^{(\gamma)}_{\eta,\kappa}(r)\frac{\mathrm{d^2}}{\mathrm{d}r^2},
\end{equation}
where the random diffusion coefficient is given by
\begin{equation}\label{eq:quadratic_variation2-gam}
    \mA^{(\gamma)}_{\eta,\kappa}(r) = \dfrac{1}{2} \int_\R \left[\phi_\eta(r-z) - \phi_\eta(-z)\right]^2 \mu_\eta(dz) + \kappa.
\end{equation}
The development leading to  
Eqs.~\eqref{eq:quadratic_variation2-gam0}--\eqref{eq:quadratic_variation2-gam}  
evidently parallels that of \cref{sec:Kraichnan}, which led to  
Eqs.~\eqref{eq:quad-gen} and \eqref{eq:coefficient kraichnan}.  
However, the current construction yields the random differential operator  
\eqref{eq:generator_MF_Kraichnan} as the infinitesimal generator of  
the regularized separation process, governing its statistics under  
quenched realizations of the field $\gf_\eta$.  
Moreover, the associated diffusion coefficient \eqref{eq:quadratic_variation2-gam} 
now involves the regularized GMC measure \eqref{eq:gmc_regularized}, 
and is thus a frozen-in-time  
random field depending on both the realization of $\gf_\eta$ and the  
intermittency parameter $\gamma$.
\cb
%\cm
%In the Kraichnan flow theory described in \cref{sec:Kraichnan},  the statistics of inter-particle separation are governed by  the second-order structure function of the driving velocity field (see \eqref{eq:quad-gen}).  Somehow, this intuitive description still holds, as  the diffusion coefficient $\mA^{(\gamma)}_{\eta,\kappa}(r)$  can be interpreted as the \textit{quenched} second-order structure  function of the multifractal field $U^\eta$  --- that is, the spatial part of the expected value of $(U^\eta(r,dt)-U^\eta(0,dt))^2$  conditioned on fixed realizations of $\gf_\eta$. The fact that this coefficient is random however leads to not-so-intuitive behaviours.
%\color{red} REMOVE because confusing?
%\cb

\subsection{The unregularized multifractal separation process}
\label{ssec:def-unregMF}

In addition to being defined by the SDE \eqref{eq:ql_model-gamma},  
the regularized separation process in our multifractal Kraichnan model 
can also be realized as a time-changed Brownian motion,  
via a construction similar to that in \eqref{eq:R=BM} and \eqref{eq:dds},  
but with the clock process given by  
\begin{equation}
    \mC_{\eta,\kappa}(t) = \int_0^t \dfrac{ds}{2 \mA^{(\gamma)}_{\eta,\kappa}(B_s)}.
\end{equation}
This integral expression remains well-posed when $\eta = \kappa = 0$  
and defines a clock process with non-trivial properties.  
Indeed, the unregularized diffusion coefficient is naturally interpreted  
as the chaos integral
\begin{equation}\label{eq:unregularized diffusion coefficient}
    \mA^{(\gamma)}(r)\equiv \mA^{(\gamma)}_{0,0} =  \dfrac{1}{2}\int_\R \left[\phi(r-z) - \phi(-z)\right]^2 \mu(dz) ,\quad r\neq 0,
\end{equation}
where $\gmc$ is the 
GMC measure given by \eqref{eq:gmc2}.
The expression above defines an almost surely finite 
random variable with expected value 
\begin{equation}\label{eq:mean A}
    \E \mA^{(\gamma)}(r) = 
    \dfrac{1}{2}\int_{\R} \left(\phi(r-z)-\phi(-z)\right)^2\,dz
    \approx r^\xi.
\end{equation}
Consequently, the unregularized clock process
\begin{equation}\label{eq:unregularized clock}
    \C(t)\equiv\C_{0,0}(t) = \int_0^t \dfrac{ds}{2\mA^{(\gamma)}(B_s)},
\end{equation}
is almost surely finite for times shorter 
than the first hitting time of the origin by $B(t)$. %
This provides a framework for analyzing the separation  
process in unregularized multifractal Kraichnan flows  
and motivates the following definition.
\begin{definition}[Multifractal Kraichnan separation process]\label{thm:multifractal kraichnan}
    \label{def:time_representation}
    We define the unregularized multifractal Kraichnan separation process
    as the time-changed Brownian motion
    $$
        R(t)=B(\tau(t)), \quad \quad t\geq0,
    $$
   
   \noindent where $\tau(t) = \inf\{s\geq0 : \C(s) > t\}$, with
   $\C$ given by \eqref{eq:unregularized clock}.
\end{definition}

As a direct consequence of Theorem 16.51 in \cite{breiman1992probability},  
we obtain the Markov property for the multifractal Kraichnan separation process  
as soon as the field $\gf$ is quenched,  along with sample path continuity 
for every typical realization of $\gf$ and $B$ --- at least up to first collision time. This is formally stated as 
\begin{theorem}\label{thm:mfk is markov}
    $\P^\gf$-a.s., up to first hitting time of $B_t$ to the origin,
    the multifractal Kraichnan separation process 
    is a stationary strong Markov process,
    with continuous sample paths and speed measure given by
    \begin{equation}\label{eq:speed2}
        \frak{m}^{(\gamma)}(dr)= \frac{dr}{2\mA^\gamma(r)}.    
    \end{equation}
\end{theorem}
To define the separation process beyond the first collision time, one needs to study the nature of the origin as a boundary point --- similar to the study of \cref{ssec:phasesmonofractal}.

%%%%%%%%%%%%%%%%%%%%%%%%%%%%%%%%%%%%%%%%%
%%%%%%%%  4: PHASE TRANSITIONS
%%%%%%%%%%%%%%%%%%%%%%%%%%%%%%%%%%%%%%%%%
\section{Phase transitions}
\label{sec:mfk-phases}
The different regimes of the multifractal Kraichnan model  
can be analyzed by applying Feller's boundary theory,  
in the same manner as in Section~\ref{sec:Kraichnan}.  
As before, we restrict the multifractal Kraichnan process  
to the positive half-line, and express the resulting speed measure as  
\begin{equation}\label{eq: full speed}  
    \frak M^{(\gamma)}(dr) = \mathds{1}_{(0,\infty)}(r)\, \frak{m}^{(\gamma)}(dr)  
    + \frak M_0^{(\gamma)} \delta(dr),  
\end{equation}  
where $\frak{m}^{(\gamma)}$ is given by \eqref{eq:speed2}.  
In the same vein as the monofractal Kraichnan theory, the different  
phases of the Lagrangian flow in the multifractal case  
are also determined by the boundary type at $r = 0$.  
This is governed by the behavior of the speed measure  
in neighborhoods of the origin, which, in turn, depends on the  
asymptotic behavior of $\mA^{(\gamma)}(r)$ as $r \to 0$  
(see~\cref{sec:Kraichnan}).

\subsection{Heuristics}\label{sec:heuristic}
To  gain intuition  
about the small-scale power-law behavior of $\mA^{(\gamma)}(r)$, let us first describe a crude 
scaling argument which, as we shall see,  turns out to give the correct result.
Expression  \eqref{eq:unregularized diffusion coefficient} suggests the identification
$\mA^{\gamma}(r) \sim \left( U(r) -U(0) \right)^2$. 
The  refined self-similarity phenomenology of Eq.~\eqref{eq:rss} then leads to the estimate
\be
	\label{eq:rssA}
	\mA^{\gamma}(r)  \simeq r^\xi \mu_r(0),
\ee
where we recall that $\mu_r(0) = r^{-1}\mu([0,r])$ denotes the local average of the GMC. 
Upon accepting that the  $\mu_r$ behaves 
like the $r$-regularized GMC defined in  \eqref{eq:gmc_regularized}  
(see \cite{rhodes2014gaussian}),  the
substitution $\eta\mapsto r$ in Eq.~\eqref{eq:gmc_regularized} yields 
\be
    \label{eq:heuristics}
    \gmc_r \widesim r^{ 2\gamma^2} e^{2\gamma\gf_r(0)} \widesim  r^{2\gamma^2},
    \quad \text{as} \quad r\to 0.
\ee
The second estimate follows from  the blunt estimate  
$\gf_r(0)\sim \sqrt{\log\frac{1}{r}}$ \cite{berestycki2024gaussian}, 
disregarding the random nature of the Gaussian field $\Gamma$ by replacing it with its typical behavior. 
Combining   Eq.~\eqref{eq:heuristics} and the refined-similarity estimate \eqref{eq:rssA} yields $\mA^{(\gamma)}(r) \widesim r^{\xi +2\gamma^2}$. 
%\begin{equation}
 %   r^{2\gamma^2} e^{2\gamma\sqrt{\log\frac{1}{r}}}
 %   \widesim r^{2\gamma^2},
 %   \quad \text{as} \quad r\to 0.
%\end{equation}
%From \eqref{eq:unregularized diffusion coefficient}, we therefore expect the following scaling behavior of $\mA^{(\gamma)}(r)$ 
%at small distances
%\begin{equation}\label{eq:scalingA1}
 %   \begin{split}
 %   \mA^{(\gamma)}(r) &\widesim  \int_{[0,r]} \left[\phi(r-z) - \phi(-z)\right]^2 \mu(dz)\\
 %   &\widesim r^{\xi}  \gmc_r\\
 %   &\widesim r^{\xi +2\gamma^2},  \quad \text{as} \quad r\to 0,
 %   \end{split}
%\end{equation}
%where we used the Hölder property of the RV kernel
%$|\phi(r-z)-\phi(-z)| \sim |r|^{\xi/2 -1/2}$. 
In turn, this  suggests  that the speed measure defined by Eq.~\eqref{eq:speed2}  behaves as
\begin{equation}\label{eq:speed_limit2}
        \frak m^{(\gamma)}(dr)\sim r^{-\xi-2\gamma^2}dr,
\end{equation}
 for $r$ small enough. 
In other words, we expect the speed measure of the multifractal Kraichnan separation process  
to behave similarly to that of the one-dimensional Kraichnan model with effective roughness exponent  
$\xi_{\text{eff}} = \xi + 2\gamma^2$,  
yielding a stochastic/deterministic dichotomy at $\xi_{\text{eff}} = 1$  
and a colliding/non-colliding dichotomy at $\xi_{\text{eff}} = 2$.

\subsection{Analysis of the diffusion coefficient}\label{sec:moments of A}

The heuristic behavior \eqref{eq:speed_limit2} can be explicitly tied to   
the scaling exponents of the small-$p$ moments of the diffusion coefficient $\mA^{(\gamma)}(r)$.
From expression \eqref{eq:unregularized diffusion coefficient}, we expect $\mA^{(\gamma)}(r)$ to exhibit a  
multifractal power-law spectrum similar to that of the driving  
multifractal velocity field \eqref{eq:uscaling}, namely, 
\begin{equation}
    \E^\gf\left[ \left( \mA^{(\gamma)}(r) \right)^p \right]  
    \sim \E\left[ \left( U(x + r) - U(x) \right)^{2p} \right]  
    \sim |r|^{\zeta_U(2p)},
\end{equation}
with scaling exponents prescribed by \eqref{eq:zeta_U}, as illustrated  in Fig.~\ref{fig:numerical evidence A}.
This intuition is substantiated by the following 
\begin{proposition}[Scaling along positive integers]\label{thm:power-law spectrum A}
    Let $p$ be a positive integer such that 
    $p < \frac{2\xi}{\gamma^2}$. 
    Then, there exists a constant $C>0$, independent of $r$, 
    such that for  $|r|\leq L$, we have
    \begin{equation}\label{eq:power_law spectrum A}
        \E^\gf\left(\left(\mA^{(\gamma)}(r)\right)^{p}\right)\leq C \left(\frac{|r|}{L}\right)^{\zeta_\mA(p)},
    \end{equation}
    with scaling exponents 
    \begin{equation}\label{eq:zetaA2}
        \zeta_{\mA}(p)=\zeta_U(2p)=\left(\xi+2\gamma^2\right)p -2\gamma^2p^2.
    \end{equation}
\end{proposition}
\cref{thm:power-law spectrum A} is not fully satisfactory as it characterizes the scaling exponents of the diffusion coefficient $\mA^{(\gamma)}(r)$ for positive integer orders only. %However, it is not sufficient to fully prescribe the small-scale asymptotics \eqref{eq:scalingA1} and the subsequent speed measure asymptotics \eqref{eq:speed_limit2}. 
What we need, however  is for the scaling law \eqref{eq:power_law spectrum A}  to hold for arbitrarily small $p$ and not just for positive integers. We do not see any fundamental reason for the scaling law 
\eqref{eq:power_law spectrum A} to break down for non-integer orders, but
extending the region of validity of Eq.\eqref{eq:power_law spectrum A} to positive and negative real values of $p$ proves technically very challenging and beyond the technical scope of this work.  We therefore choose to trust the numerical evidence  in Fig. \ref{fig:numerical evidence A}, which strongly suggests that \eqref{eq:power_law spectrum A}  holds for non-integer orders, including small negative values of $p$,
and  henceforth work under the following 
\begin{assumption} \label{thm:assumption}
    There exists $\delta>0$ such that, for $p \in (-\delta,\delta)$
    and for $|r|\leq L$, we have 
    \begin{equation}
        \E^\gf\left(\left(\mA^{(\gamma)}(r)\right)^{p}\right)\leq C \left(\frac{|r|}{L}\right)^{\zeta_\mA(p)},
    \end{equation}
    where $C>0$ is a constant independent of $r$,
    and $\zeta_\mA$ is given by \eqref{eq:zetaA2}. 
\end{assumption}

\begin{figure}[h]
    \includegraphics[width=0.49\textwidth]{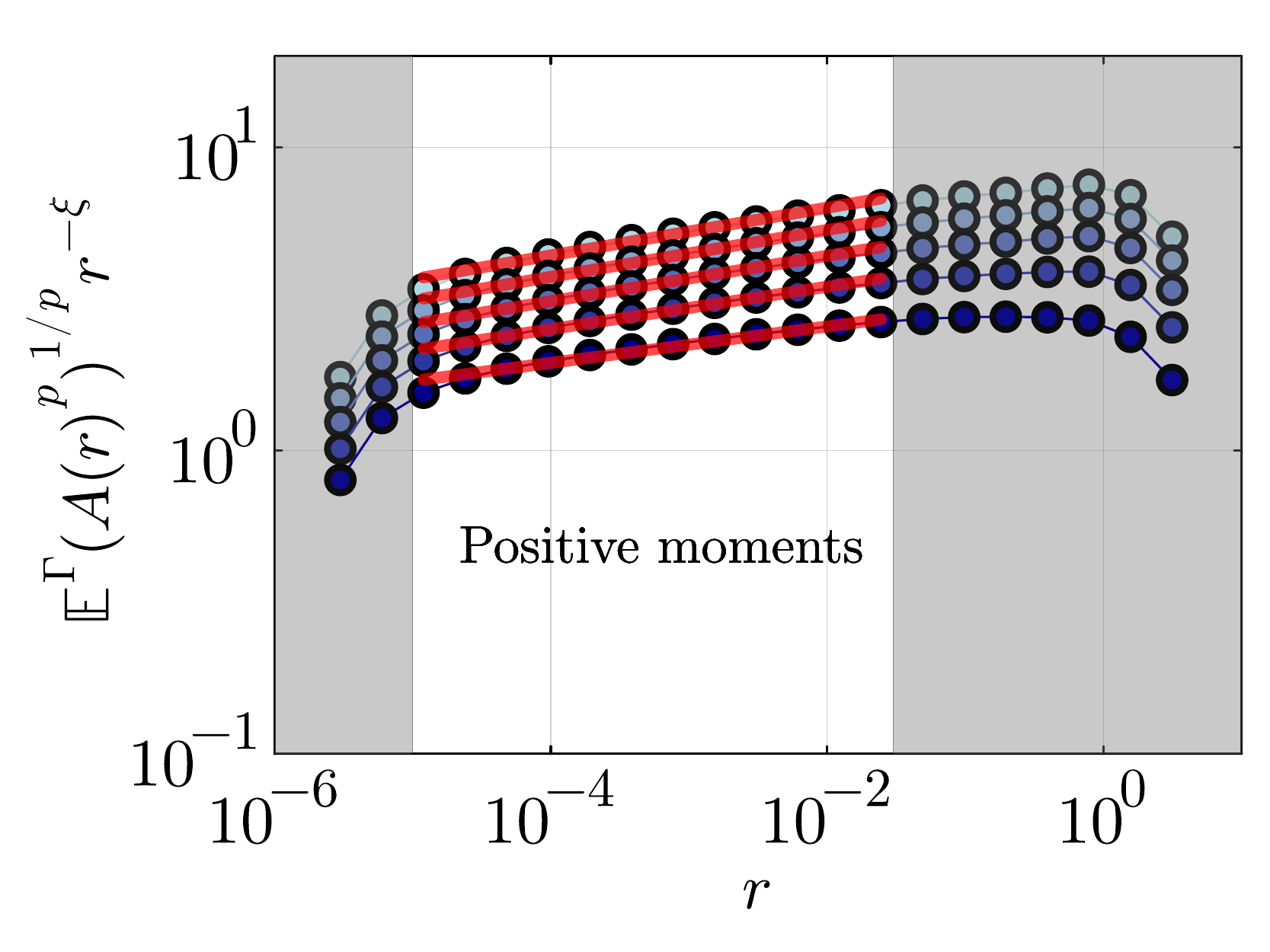}
    \includegraphics[width=0.49\textwidth]{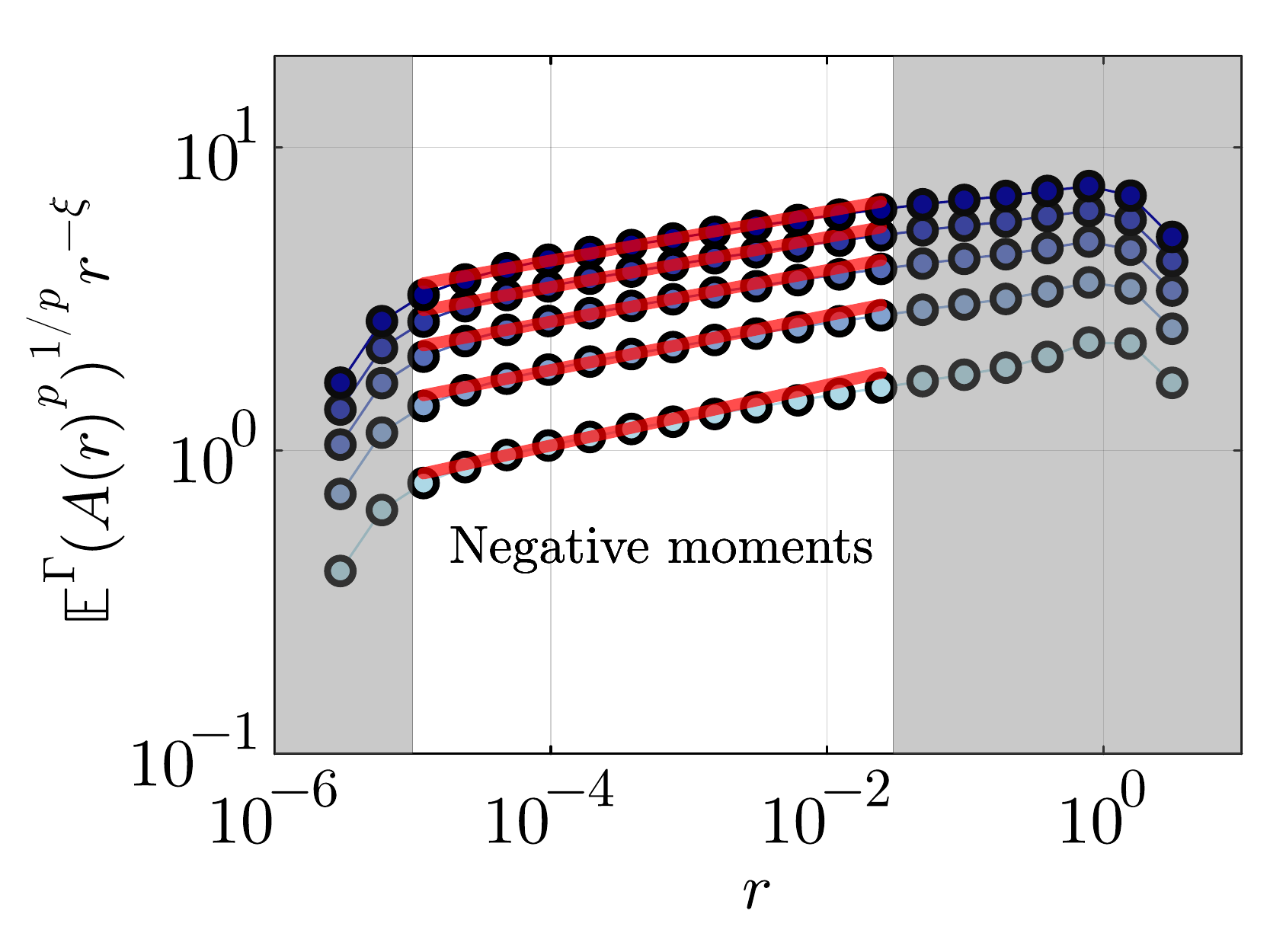}
        \caption{Left: Compensated moments of the diffusion coefficient 
        $\E^\gf\left(\mA^{(\gamma)}(r)^p\right)^{1/p}r^{-\xi} \propto r^{\zeta_A(p)/p-\xi}$
        obtained by Monte Carlo sampling using $2^{22}$ grid points, 
        for $\gamma=0.2$, $\xi =2/3$ and 
        $p=1/64, 1/32, 1/16, 1/8,1/4$ (from top to bottom). 
        Right: Same as the left panel, but 
        for $p=-1/64, -1/32, -1/16, -1/8, -1/4$ (from top to bottom). 
        }
        \label{fig:numerical evidence A}
\end{figure}

Under the \cref{thm:assumption}, we are able to deduce the scaling behavior 
predicted heuristically by the scaling argument in \cref{sec:heuristic},
which in particular implies that, 
around the origin, 
$\mA^{(\gamma)}(r)$ is 
$H$-Hölder continuous for every $H<\xi+2\gamma^2$ and 
for every typical realization of the GMC.
\begin{proposition}\label{prop:Aestimate}
    Suppose that \cref{thm:assumption} holds.
    Let $\eps>0$. 
    Then, $\P^\gf$-a.s., 
    \vspace{0.2cm}
    \begin{equation*}
        r^{\xi+2\gamma^2+\eps} \leq \mA^{(\gamma)}(r)\leq r^{\xi+2\gamma^2-\eps}
        \quad \text{as} \quad r\to0.    
    \end{equation*}
\end{proposition}

The proofs of \cref{thm:power-law spectrum A} and \cref{prop:Aestimate} can be found  in Appendix \ref{sec:proofs}.

\subsection{Quenched phase transitions}

From the estimate given by \cref{prop:Aestimate}, we obtain the classification of the boundary point 
$r=0$ for the multifractal Kraichnan process, following the analysis of Subsection \ref{ssec:unregularizedKraichnan}.
The classification can be carried out by replacing $\xi$ 
with $\xi_{\text{eff}}=\xi + 2\gamma^2$ in the one-dimensional Kraichnan model, leading to 
\begin{theorem}\label{thm:multifractal phases}
    Suppose that the \cref{thm:assumption} holds. 
    Let $\xi\in (0,2]$ and $\gamma \in [0,\sqrt 2/2)$.
	Then, for the unregularized multifractal Kraicnan separation process in \cref{thm:multifractal kraichnan}, the origin is,
    $\,\,\P^\gf$-a.s.,
    \be
	\nonumber
        (I)  \text{ regular if } \xi< 1 - 2\gamma^2 ,\quad
        (II)\text{  exit if } 1 - 2\gamma^2\leq  \xi< 2 - 2\gamma^2 ,\quad
       (III) \text{ natural if } 2 - 2\gamma^2\leq\xi.
   \ee
The phase diagram is represented in the left panel of Fig. \ref{fig:phases}. 
\end{theorem}

\cref{thm:multifractal phases} rigorously establishes a smoothing  
effect of intermittency on Lagrangian dispersion.  
For small values of $\xi$ and $\gamma$, the origin is a regular  
boundary point. In this regime, similarly to the (monofractal)  
Kraichnan model, the fate of the Lagrangian flow depends on the  
weight $\frak M_0^{(\gamma)}$. 
As discussed in \cref{ssec:vanishingreg}, this quantity
not only heavily depends on the regularization scheme employed but
also on the precise manner in which the regularization is removed.
If viscosity is removed first, reflecting boundary condition  
is selected, whereas removing thermal noise first leads 
to absorbing boundary condition.

For $\xi<1$,  the multifractal Kraichnan model  
transitions into the exit phase (II)  as $\gamma$ increases,  
In this phase, particles coalesce upon collision regardless of the limiting procedure.
This universal behavior with respect to coalescence  in the multifractal Kraichnan model for $\xi<1$ is a manifestation of the smoothing effect of intermittency.
For $1<\xi<2$, we observe a transition  to the natural phase (III) by increasing $\gamma$.  
In this region, particles neither collide with one another  nor branch out when starting from the same initial position.
This behavior is typically associated with smooth flows, where the driving velocity field is Lipschitz continuous and a unique  solution to the flow equations exists.  The fact that the natural phase extends for $\xi<2$,  at the cost of increasing $\gamma$,  suggests that the flow equations driven by rough intermittent  velocity fields ($1<\xi<2$) admit a unique solution,  provided that $\gamma$ is sufficiently large.   In other words, the smoothing effect  seems to induce a phenomenon of \emph{regularization by intermittency}. 

\begin{figure}[h]
    \includegraphics[width=0.49\textwidth]{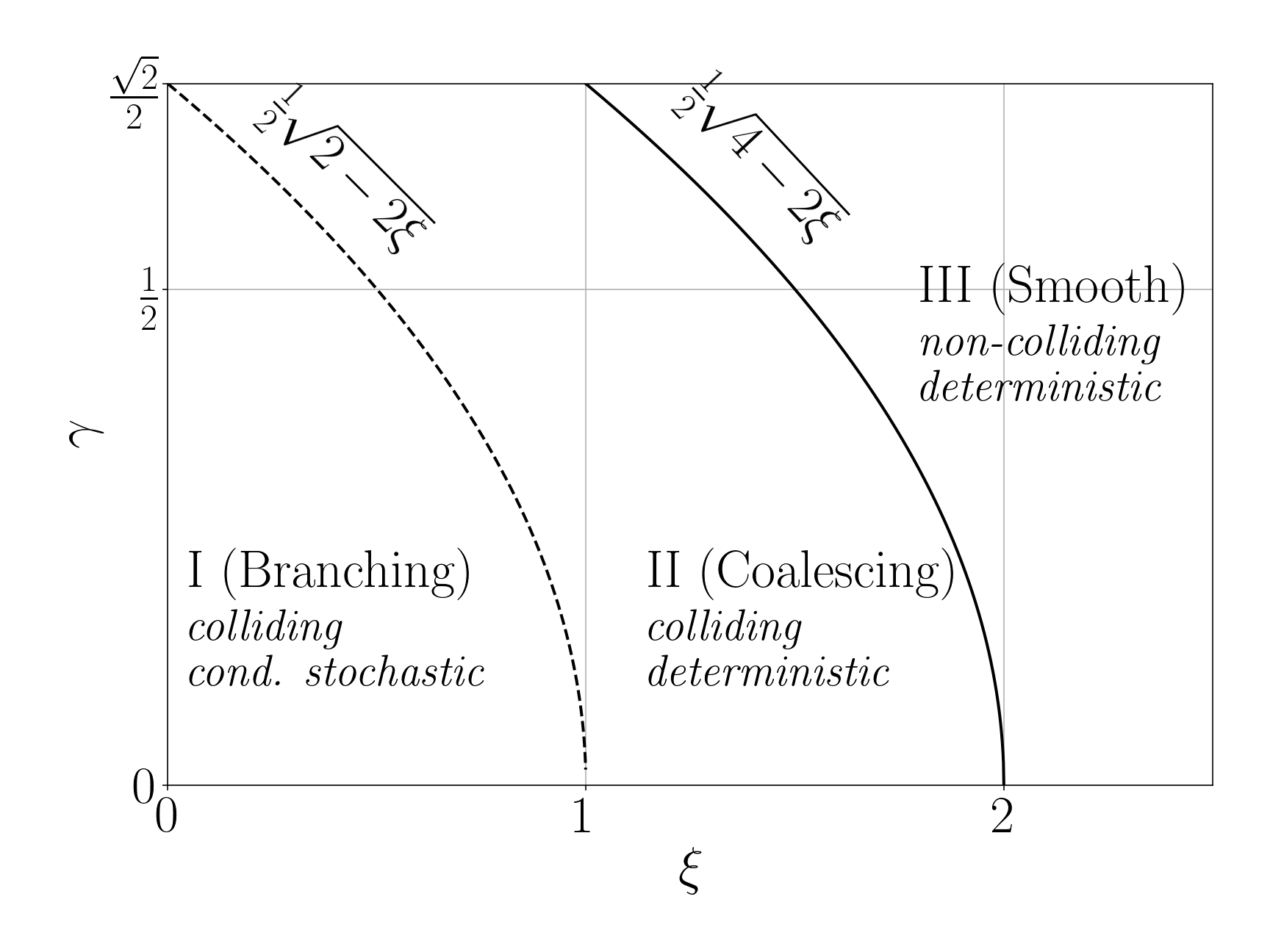}
    \includegraphics[width=0.49\textwidth]{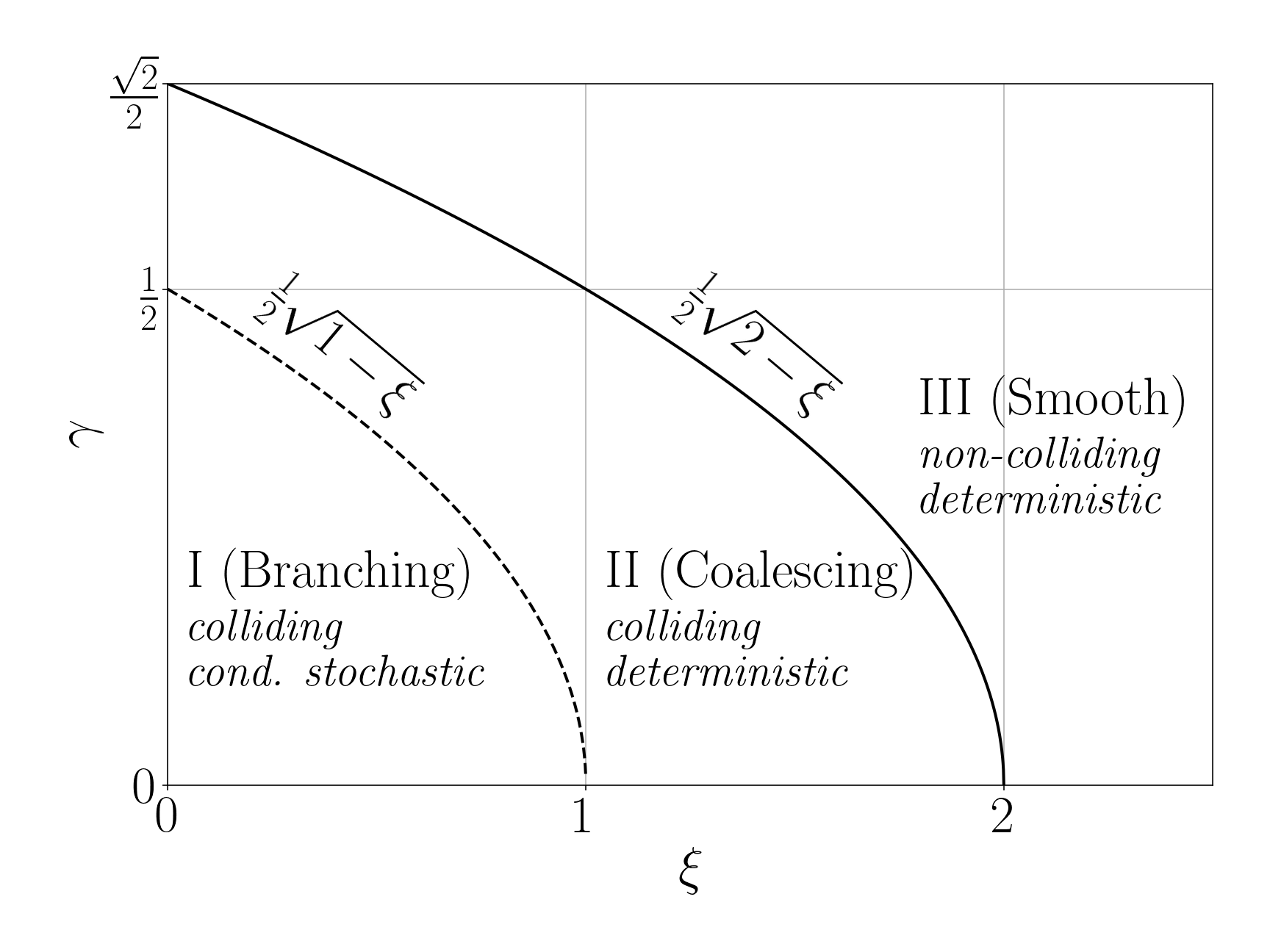}
    \caption{Phases of the Lagrangian flow in the multifractal Kraichnan model. Left: Quenched setting of Theorem \ref{thm:multifractal phases}. Right: Annealed setting of \S \ref{ssec:annealed}.}
    \label{fig:phases}
\end{figure}

\subsection{Annealed phase transitions}\label{ssec:annealed}

In a previous numerical  paper \cite{considera2023spontaneous},  we studied a finite-dimensional approximation  of the multifractal Kraichnan model, obtained by interpreting the advection of particles  as an interaction through a random pairwise potential.  
The use of this simplified model allowed for a mean-field Ansatz,  which ultimately led to the conclusion that  phase transitions   are prescribed by the ``mean-field exponent''  $\xi_{mf} = \xi + 4\gamma^2$, different from the quenched exponent  $\xi+2\gamma^2$ of Theorem \ref{thm:multifractal phases}.
We however point out that the exponent $\xi_{mf}$ 
can be interpreted as the relevant exponent driving 
the phase transition of the separation process in the \emph{annealed setting}.
The annealed setting is obtained by letting the separation evolve 
in an averaged environnment.

The annealed separation process is defined as the 
time-changed Brownian motion parametrized by the annealed clock process
\be\label{eq:unregularized clock2}
    \C_{mf}(t) = \mathbb E^\Gamma \int_0^t \dfrac{ds}{2 \mA^{(\gamma)}(B_s)} =  \int_{\mathbb R}\ell(y,t) \E^{\gf} {\frak m}^{(\gamma)} (dy) ,
\ee
where we recall that $\ell$ denotes the Brownian local time. 
The \emph{annealed speed measure} is defined by  averaging the quenched
speed measure \eqref{eq:speed2} over realizations of $\gf$
in the following manner
\be
        \E^{\gf} {\frak m}^{(\gamma)} (dr) = \E^\gf \left(\left(2 \mA^{\gamma}(r)\right)^{-1}\right)\,dr\\
        \approx r^{\zeta_\mA(-1)} dr\\
        \sim r^{-\xi_{mf}} dr.
	\label{eq:meanfieldFeller}
\ee
Note that, to obtain the scaling relation above, we have assumed that the 
scaling law in \cref{thm:power-law spectrum A} also holds for $p=-1$. 
The infinitesimal generator of the annealed separation process
can then be recovered from the speed measure  
through relation \eqref{eq:generator_speed}  as 
\begin{equation}\label{eq:annealed generator}
    \mL_{mf} = r^{\xi_{mf}}\frac{\mathrm{d^2}}{\mathrm{d}r^2}.
\end{equation}
Observe that the operator above is deterministic, with no GMC component,  
as a consequence of averaging out its multifractal contribution.  
In the annealed setting, the multifractal corrections are thus fully prescribed  by the  
(effective) mean-field exponent $\xi_{mf}$. Substituting $\xi$ by $\xi_{mf}$ in 
Theorem \ref{thm:phase} implies a 
stochastic/determistic dichotomy at $\xi_{mf} =1$,
and a colliding/non-colliding dichotomy at $\xi_{mf}=2$. 
This leads to the \emph{annealed} diagram represented in the  right panel of Fig.\ref{fig:phases}. \cb The net effect of the averaging procedure is to make the smoothing effect of intermittency more pronounced, with the transitions  I$-$II and II$-$III occurring at smaller values of $\gamma$.
Moreover, the annealed Lagrangian flow can transition  from regular to natural for arbitrarly small values of $\xi$.
In this sense, averaging the speed measure over realizations of $\gf$
provides an additional smoothing mechanism.

The relevance of the mean-field exponent $\xi_{mf}$ was previously pointed out  in  \cite{considera2023spontaneous} using a heuristic approach, which required substituting the random fields by random potentials --- in the spirit of the Bessel shortcut of \cref{ssec:Besselshortcut}.
%\cite{considera2023spontaneous} prescribe the phases of the annealed separation process. This can be checked upon  substituting $\xi$ by $\xi_{mf}$ in Theorem \ref{thm:phase}.
The approach developed here provides a clean way to retrieve the predictions of \cite{considera2023spontaneous},  
without the need for such (a priori unjustified) substitution.  As a side note, we point out that the (environmental) averaging over the GMC measure leading to the annealed process
%The approach outlined in this subsection bears strong 
bears a formal similarity
to stochastic homogenization theory for random elliptic operators.  
In that context, one considers a diffusion operator with random coefficients  
and seeks to derive an effective, deterministic operator  
that governs the large-scale behavior of the underlying  
diffusion process in the so-called ``homogenization limit'',  
that is, in the limit where the scale of heterogeneity tends to zero.  
In the particular case of one-dimensional operators in non-divergence form,  
the effective diffusion coefficient is given by the harmonic average of the  
initial, environment-dependent random coefficient \cite{papanicolaou1982diffusions,papanicolaou1995diffusion,pavliotis2008multiscale}.
% --- namely,  
%the inverse of the expected value of its reciprocal
%\cite{papanicolaou1982diffusions,papanicolaou1995diffusion}
%(see also \cite{pavliotis2008multiscale} for an example in the periodic setting).
This is in close correspondence with the calculations that yielded  
the generator~\eqref{eq:annealed generator}.  
It is thus tempting to view~\eqref{eq:annealed generator}  as the homogenized version of the quenched generator associated to the speed measure~\eqref{eq:speed2},  governing the statistics of the multifractal Kraichnan separation process  
in a diffusive scaling limit of the type $\lim_{\varepsilon \to 0} \varepsilon R(t/\varepsilon^2)$. Such connection could prospectively hold in phase (I), in the absence of trapping effects, but its rigorous substantiation lies beyond the scope of the present paper. 
%Nevertheless, we would like to emphasize that the homogenization of random GMC operators  
%poses an interesting and challenging problem in stochastic homogenization theory,  
%with potential implications for understanding the effective impact of small-scale intermittency  
%on the large-scale motion of fluid particles in turbulent flows.
\cb

%%%%%%%%%%%%%%%%%%%%%%%%%%%%%%%%%%%%%%%%%
%%%%%%%%  5: MULTIFRACTAL BEHAVIOR
%%%%%%%%%%%%%%%%%%%%%%%%%%%%%%%%%%%%%%%%%

\section{Multifractal behavior}
\label{sec:multifractal}
The smoothing behavior of the Lagrangian flow in the quenched setting admits a physical interpretation in terms of the Parisi–Frisch multifractal spectrum, which, within the GMC framework, is formulated using the notion of $\alpha$-thick points.
Interpreting the phase diagram of \cref{thm:multifractal phases}  through the lens of the multifractal formalism sheds some light on the  
interplay between Eulerian roughness and the different phases of the Lagrangian flow.

\subsection{$\alpha$-thick points and Parisi-Frisch multiscaling}
\label{ssec:thickpoints}
\begin{figure}
 	\includegraphics[width=0.49\textwidth]{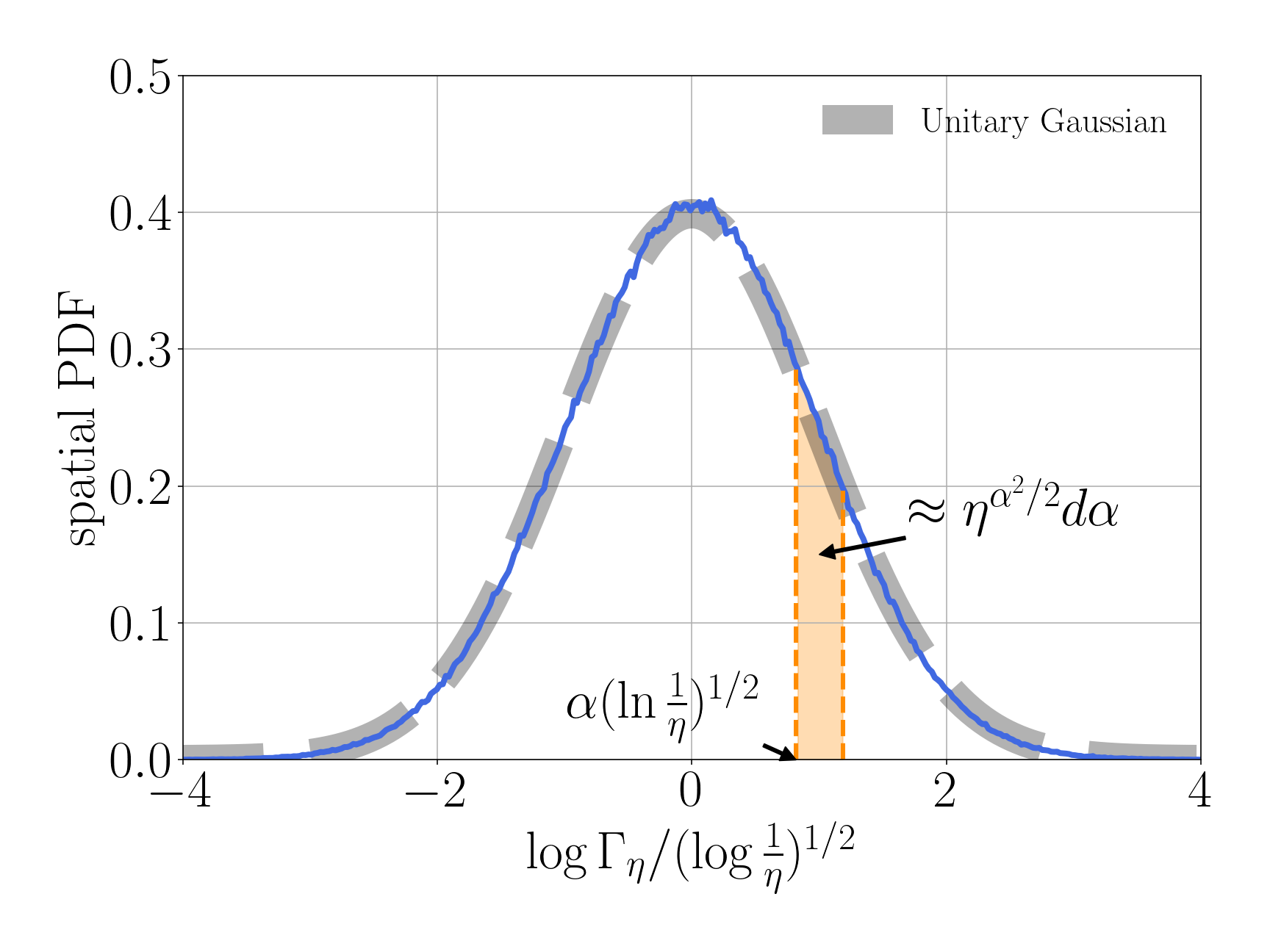}
 	\includegraphics[width=0.49\textwidth]{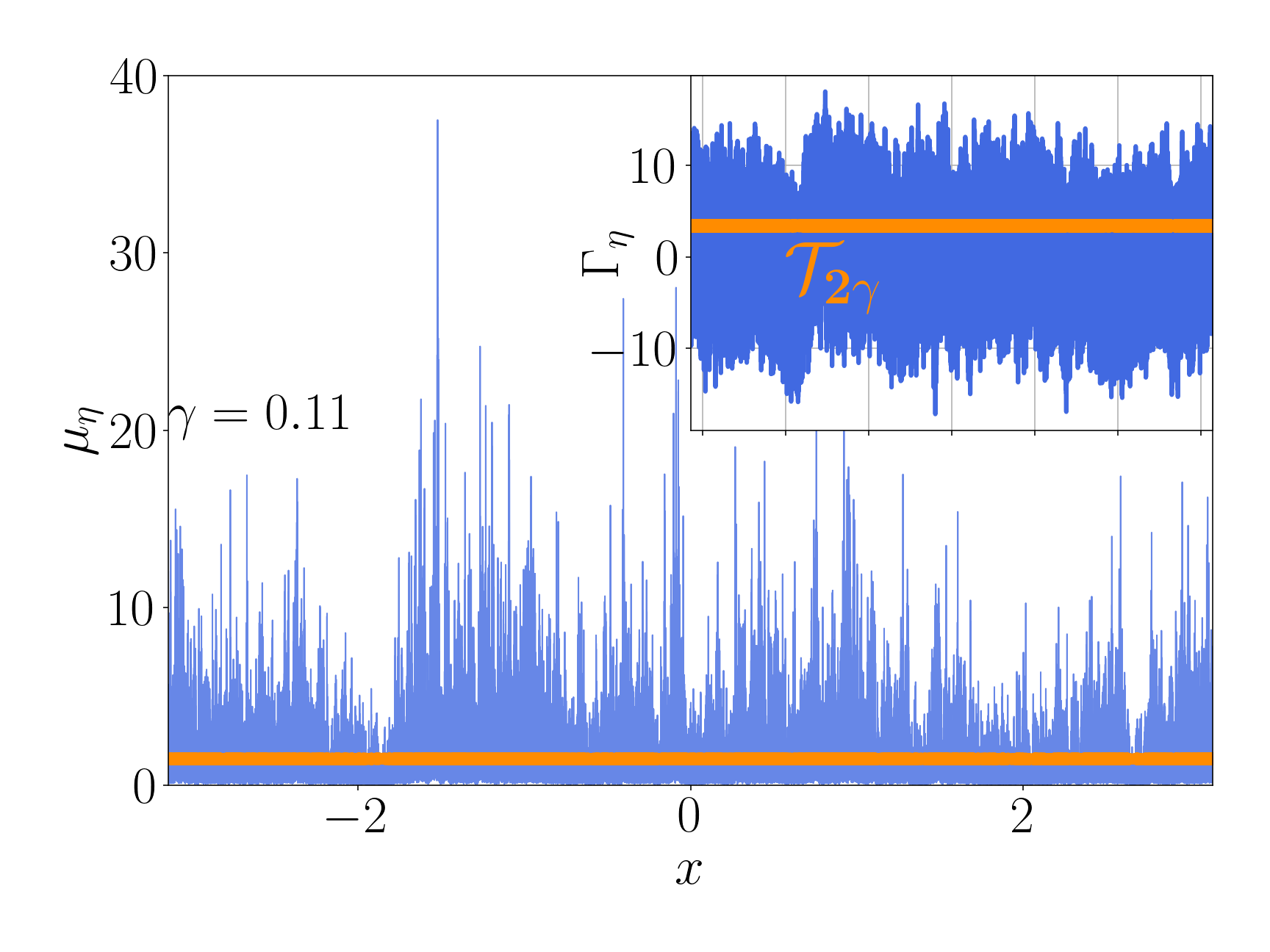}
	\caption{Left: Spatial distribution of one realization of the  $\epsilon$-GMC on $(-\pi,\pi)$. The shaded area represent contribution from  the $\alpha$-thick-point for $\alpha =2\gamma$.  Right: $\mu_\eta,\Gamma_\eta$ together with their $2\gamma$- thick points in orange.}
	\label{fig:2}
\end{figure}
Beyond the  quadratic behavior of the mass scaling exponents \eqref{eq:power_law_spectrum_gmc},  another facet of multifractality in GMC theory is formalized through the notion of $\alpha$-thick points 
of the generalized field $\gf$. More precisely,  
for $\alpha\in \R$,
we define the set $\T_\alpha$ of $\alpha$-thick points of $\gf$ as  
\begin{equation}\label{eq:thickpoints}
    \T_\alpha =\left\{x\in [0,L] \,\,\, : \,\,\, \lim_{r\to 0}\frac{\gf_r(x)}{-\log(r)}=\alpha\right\}.
\end{equation}
%\be
%	\label{eq:thickpoints}
%	T_\alpha = {\lim \inf}_{r \downarrow  0 }T^{(r)}_\alpha , \quad 	T^{(r)}_\alpha = \left\lbrace x  \in \mathbb R, \quad 
%\dfrac{\Gamma_r(\bx)}{-\log r} = \alpha \right\rbrace,
%\ee
As mentioned in \cref{sec:heuristic}, 
$\gf_r$ typically behaves as $\sqrt{\log(1/r)}$ 
as $r$ goes to $0$.
As a result, typical behavior is described by $0$-thick points,
whereas non-typical behavior is characterized by 
$\alpha$-thick points with $\alpha \neq 0$.
In this latter case, the set of $\alpha$-thick points 
corresponds to regions in space where the field $\gf_r$ is 
unusually large. 
The $\alpha$-thick points are responsible for the scaling properties of the GMC:
if $x \in \T_\alpha$, then we have local scaling property $\Gamma_r(x)  \asymp r^{-\alpha}$. 
This entails that the corresponding local GMC averages scale as 
$\left.\mu_r(x)\right|_{x \in \T_\alpha} \asymp r^{2\gamma^2-2\gamma\alpha}$.

The sets $\T_\alpha$ are known to form a disjoint cover of the underlying  
space, serving as carriers for GMC measures  \cite{rhodes2014gaussian,kahane1985chaos}.  
More specifically, $\gmc\left(\T_{\pm 2\gamma}^c\right) = 0$,  that is, the sets $\T_{\pm 2\gamma}$ give full mass to $\gmc$.  
%As a corollary, if $\gamma$ and $\gamma'$ are such that  $\gamma \neq \gamma'$, then the GMC measures $\gmc_\gamma$ and  
%$\gmc_{\gamma'}$ are mutually singular.  
In addition, upon identifying the Lebesgue measure as  a GMC measure with vanishing intermittency parameter,  
one deduces that the set of typical points $\T_0$ gives full  
mass to it, and moreover, $\text{Leb}(\T_\alpha) = 0$ if $\alpha \neq 0$.
%Focusing on  the thick points within a  finite interval, say  $(0,1)$, the set $T_\alpha \cap (0,1)$ 
The sets $\T_\alpha$, however, have nontrivial Hausdorff dimension  
given by $D_\alpha = \left(1-\frac{\alpha^2}{2} \right)_+$ \cite{hu2010thick,rhodes2016lecture,berestycki2024gaussian}
--- a feature qualitatively  illustrated in Fig.~\ref{fig:2}.

These properties suggest that, in a given realization of the GMC, the  mass moments obtained from  space-averaging of the $r$-mollified GMC can be computed as 
\be
	\label{eq:PFrisch}
	\left\langle \mu^p_r(x)\right \rangle_{x\in (0,L)}  \asymp \int_{-\sqrt 2}^{\sqrt 2} d\alpha\; r^{1-D_\alpha}r^{2p\gamma^2 - 2p\gamma\alpha} \asymp r^{\tilde\tau_p},\quad \quad \tilde \tau_p = \inf_{\alpha\in (-\sqrt 2, \sqrt 2)}\left\lbrace \dfrac{\alpha^2}{2} + 2p\gamma\left(\gamma-\alpha\right) \right\rbrace,
\ee
prescribing the exponents of  $L_p$ norms as Legendre-Fenchel transforms \cite{touchette2005legendre}: Such formula is the classical Parisi-Frisch multifractal framework \cite{frisch1985singularity} known in turbulence studies, where $D_\alpha$ is known as the singularity spectrum. %, and which we illustrate in Fig.~\ref{fig:3}.
The exponents $\tilde \tau_p$ coincinde with $\tau_p$ of Eq.~\eqref{eq:power_law_spectrum_gmc} provided $|p|<\hspace{-0.15cm}\sqrt{2}/\gamma$ and otherwise exhibit linear scaling with $p$. The difference with Eq.~\eqref{eq:power_law_spectrum_gmc} comes from the Parisi-Frisch formula being 
a ``quenched" formula involving non-trivial 
coverings of the space, 
while Eq.~\eqref{eq:power_law_spectrum_gmc} is 
an ``annealed" formula that averages over the 
environment. 
A clean mathematical formulation of the
Parisi-Frisch multifractal formalism 
\eqref{eq:PFrisch} is the scope of the multifractal analysis of measures, 
leading in particular to the so-called Khnizhnik-Polyakov-Zamolodchikov (KPZ) 
relations in 2D LQG \cite{rhodes2014gaussian,berestycki2024gaussian}. 

At the level of the field \eqref{eq:velocity},
the refined self-similarity heursitics 	\eqref{eq:rss}
suggests that  the sets $\T_\alpha$ support the scaling behavior 
$U(x+r)-U(x) \sim  \mu_r(x)^{1/2}|_{x\in \T_\alpha} r^{\xi/2}\sim   r^H$ with the  H\"older exponent 
$H=\xi/2+\gamma^2-\gamma \alpha$. 
As such, the multifractal random field \eqref{eq:velocity} intertwines a wide range of singular structures characterized by exponents $H \in [H_{\min}\,,\,H_{\max}]$.
The most probable exponent is given by $\overline H = \xi/2 + \gamma^2$, corresponding to $\alpha=0$,
and with space-filling fractal dimension $D_0=1$. 
The least and most singular structures correspond to, respectively,  
$H_{\max} = \xi/2+\gamma^2+\gamma\sqrt 2$ and $H_{\min} = \xi/2+\gamma^2-\gamma\sqrt 2$, 
each supported on a set of vanishing Hausdorff dimension $D_{\pm \sqrt 2}=0$.

\subsection{Lagrangian vs Eulerian roughness}
\label{ssec:LagvsEule}
As the intermittency parameter $\gamma$ increases,  the lowest H\"older exponent $H_{\min}$ decreases. 
Such maximal roughness is captured by the Kolmogorov criterion, which estimates the 
rougness parameter  of a random field with scaling exponents $\zeta_U(p)=p\left(\xi+2\gamma^2\right)-\frac{\gamma^2p^2}{2}$ as
twice the lowest H\"older exponent $H_{\min}$ through the formula
\begin{equation}\label{eq:Kolmogorov_supremum}
    \xi_{\min} = 2\sup_p \frac{\zeta_U(p)-1}{p},\quad \text{yielding}\quad     \xi_{\min}=2H_{\min}=\xi +2\gamma^2-2\sqrt{2}\gamma.
\end{equation} 
For $\gamma=0$,  the expression above reduces to 
$\xi_{\min}=\xi$ and we recover the monofractal case. For $0<\gamma<\sqrt 2/2$, as allowed by GMC theory, one has $\xi_{\min}<\xi$. This means that, in the Kolmogorov sense,  the multifractal velocity field \eqref{eq:velocity} is rougher than its monofractal counterpart.
By constrast, Theorem \ref{thm:multifractal phases} identifies the relevant parameter for the Lagrangian flow as 
\begin{equation}
    \overline \xi  =  \xi+2\gamma^2= 2 \frac{\mathrm{d}\zeta_U(p)}{\mathrm{d}p} \Big|_{p=0},
\end{equation}  
where the second identity is obtained from \eqref{eq:zeta_U} by direct calculation.
Consequently, it is not the most singular structures that govern  the Lagrangian flow, but the bulk effect of the singularity  
spectrum, characterized by the most probable H\"older exponent $\overline H = \overline \xi/2>\xi/2$.

\subsection{Multifractal Kraichnan spends most of its time at typical points}\label{sec:mfk multifractal behavior}

The relevance of the most probable exponent $\overline{H}$, rather than $H_{\min}$,  
as the driving parameter in \cref{thm:multifractal phases}  
stems from the behavior of the Lagrangian trajectories.
In this section, we show that the one-dimensional multifractal Kraichnan
process tends to concentrate around points where the 
velocity field undergoes typical fluctuations. 
As a result, for most of their temporal evolution,  particles are carried by the most typical 
structures in the flow, characterized by $\overline{H}$.

For simplicity,
we work with the multifractal Kraichnan diffusion process 
on the bounded interval $[0,L]$. Contrary to the previous 
sections, we are now interested in the behavior in the interior 
$(0,L)$ rather than at the  boundaries.
The fundamental observation is that on $(0,L)$, the speed measure \eqref{eq:speed2} is absolutely continuous with respect  to the Lebesgue measure, allowing us to conclude that  $\frak m^{(\gamma)}(\T_\alpha) = 0$ 
if $\alpha \neq 0$. \cb
Let us comment what it means for a set to have vanishing  
speed measure. For general driftless one-dimensional diffusions,  
the density of the speed measure is inversely proportional to the  
diffusion coefficient, meaning that regions where the speed measure  
is small correspond to regions where the underlying diffusion process  
diffuses rapidly.  
Such regions are naturally identified as those in which  
the diffusion process spends little time,  
as a consequence of large fluctuations pushing the process away from them.  
Thus, the speed measure of a set can be viewed as a measure of the  
total amount of time spent by the process in the given set.
The extreme case of vanishing speed measure then corresponds  
to the situation in which the diffusion process  
spends no time in the given region.  
Applied to the present setting,  
$\frak m^{(\gamma)}(\T_\alpha) = 0$ for $\alpha \neq 0$
indicates that the multifractal Kraichnan diffusion spends no time  
at non-typical points of the field $U$ and, in fact,  
spends most of its time at typical points.
This heuristic is captured by the following 
\begin{proposition}\label{thm: mfk spends time}
The multifractal Kraichnan diffusion process 
spends Lebesgue-almost all its time in typical points of $\,\,\gf$,
that is,

$$
\P^\gf \text{-} a.s., \;\; \text{for all}\;\; r\in (0,L), \;\;
\P^B_r\text{-}a.s.,\;\; 
\text{Leb}\left(\left\{0\leq t\leq T \,\,:\,\, R(t)\in \T_0^c \right\}\right)=0,
$$
where $T$ is the first exit time of $R$ from $(0,L)$.
\end{proposition}
We refer the reader to \cref{sec:proofs} for the proof. 
The proposition above allows us to conclude that  
particles almost surely never remain in  
regions exhibiting non-typical velocity fluctuations   for an uninterrupted, strictly positive time interval.  
However, they can still occupy those regions at isolated instants of time  
before being pushed away from them by strong local fluctuations.  
Still, in this latter case, the cumulated time spent at those regions 
is almost surely zero. In this sense, particles are instantaneously 
 repelled from 
%deflected  at 
the $\alpha$-thick points of $\gf$ for  $\alpha\neq0$. 

The emergence of $\overline{H}$ 
as the critical parameter characterizing the phase transitions 
in the multifractal Kraichnan model becomes natural:  Since particles are effectively concentrated in regions  
exhibiting typical velocity fluctuations, one naturally   expects $\overline{H}$ to be the most relevant Hölder exponent  
in the Parisi-Frisch singularity spectrum, as typical fluctuations   are described by $\overline{H}$.  
\cb

%%%%%%%%%%%%%%%%%%%%%%%%%%%%%%%%%%%%%%%%%
%%%%%%%% 6: MLBM
%%%%%%%%%%%%%%%%%%%%%%%%%%%%%%%%%%%%%%%%%

\section{Analogies with Liouville quantum gravity}\label{sec:MLBM}

Similar to the Bessel shortcut discussed in \S \ref{ssec:Besselshortcut} for the monofractal Kraichan flows,
one could wonder whether the phase diagrams in Fig. \ref{fig:phases}  can be retrieved from an effective diffusive process, only this time coupled to a random geometry induced by the GMC. 
The construction of diffusion processes via time changes  
based on GMC measures has been explored previously  
in \cite{garban2016liouville, berestycki2015diffusion},  
within the 2D setting of Liouville quantum gravity (LQG).  
The canonical diffusion in a random planar geometry induced by the GMC is called Liouville Brownian motion (LBM).
In this section, we introduce a unidimimensional multiplicative version of the LBM,
which we refer to as \emph{multiplicative Liouville Brownian motion} (MLBM).
The MLBM formally solves the SDE
\begin{equation}\label{eq:MLBM}
    d\B=|\B|^{\xi/2} e^{-\gamma\gf(\B)
    + 
    \gamma^2\E(\gf(\B)^2)} d\beta,
\end{equation}
and in line with the previous sections, it is effectively defined as a time-changed Brownian motion. 
We then argue that the phases of the MLBM can be mapped onto the multifractal phase diagram  
of Fig.~\ref{fig:phases}, in both the quenched and annealed settings.  
However, its multifractal behavior differs significantly from that of the multifractal separation process.  
For readers not familiar with LBM theory, we refer to  
\cite{garban2016liouville, berestycki2015diffusion} for background material.
\subsection{Multiplicative Liouville Brownian motion (MLBM)}
Upon regularizing the GMC by changing  $\Gamma \mapsto \Gamma_\eta$ in Eq.~\eqref{eq:MLBM}, 
the  MLBM can be thought of as a limiting  solution for   following 
SDE 
\begin{equation}\label{eq:MLBM2}
    d\B_\eta= \eta^{-\gamma^2} |\B_\eta|^{\xi/2} e^{-\gamma\gf_{\eta}(\B)}
    d\beta,
\end{equation}
where $\beta$ is a Brownian motion, and $\gf_\eta$ is a smooth regularization 
of the log-correlated field $\gf$, independent of $\beta$ 
(see \cref{sec:gmc2}). 
If $\xi=0$, we obtain a  regularized LBM with parameter $2\gamma$ as defined in  \cite{garban2016liouville, berestycki2015diffusion}.
Although the equation above makes sense  
in the presence of the small-scale cutoff $\eta$,  it is not well defined directly at $\eta=0$ 
due to the logarithmic singularity   of the unregularized correlation function of $\gf$ at short distances.

%In order to define the unregularized MLBM, we follow \citet{garban2016liouville} and use time-changed Browian motions. To motivate the definitions that follow,  
Similar to the (regularized) multifractal   Kraichnan process (see \cref{ssec:def-unregMF}), the $\eta$-regularized  MLBM can be written  
as the time-changed Brownian motion  
\begin{equation}
    \B_{\eta}(t) \stackrel{\text{law}}{=} B(\tau_{\eta}(t))  
    \quad\text{for}\quad t \geq 0,
\end{equation}  
where  
\begin{equation}\label{eq:regularized_clock_MLBM}
    \tau_{\eta}(t) = \inf\{s\geq0 : \C_{\eta}(s) > t\},  
    \quad 
%    \C_{\eta}(t) = \int_{0}^{t}\eta^{2\gamma^2}    \frac{e^{{2\gamma} \gf_\eta(B(s))-2\gamma^2\E(\gf_\eta(B(s))^2)}}{|B(s)|^{\xi}}      \, ds,
    \C_{\eta}(t) = \eta^{2\gamma^2} \int_{0}^{t}
    \frac{e^{{2\gamma} \gf_\eta(B(s))}}{|B(s)|^{\xi}}  
    \, ds.
\end{equation}  
Recognizing 
the regularized GMC measure  $\gmc_\eta$ of Eq.~\eqref{eq:gmc_regularized}, the clock process can be rewritten
as
\begin{equation}\label{eq:clock_occupation}
    \C_{\eta}(t)  
%\int_{0}^{t}
%    \frac{e^{{2\gamma} \gf_\eta(B(s))-2\gamma^2\E(\gf_\eta(B(s))^2)}}{|B(s)|^{\xi}}  
 %   \, ds
    = \int_{\R} \ell(t,\rho) |\rho|^{-\xi} \gmc_\eta(d\rho),
\end{equation}  
where $\ell$ is the Brownian local time and $\mu_\eta$ is the regularized GMC measure \eqref{eq:gmc_regularized}.
 In the absence of small-scale regularization, the above expression  
is naturally interpreted as a chaos integral.  
This motivates the following definitions for the relevant objects,  
directly at $\eta=0$.

\begin{definition}[Clock process]\label{thm:def_clock_time}
    Let $\xi\in[0,2]$ and $\gamma\in\big[0,\frac{\sqrt{2}}{2}\big)$.
    We define the MLBM clock process
    as the random functional  
    $$
    \C(t)=\int_0^\infty \ell(t,\rho) \, |\rho|^{-\xi}\mu(d\rho),
    \quad t\geq 0
    $$
    where $\ell$ is the local time of Brownian motion reflected at zero and $\mu$ is the GMC measure of   \cref{sec:gmc2}. 
\end{definition}
\begin{definition}[MLBM]
    \label{def:time_representation}
    We define the multiplicative Liouville Brownian motion, 
    starting at $r>0$,
    as
    $$
        \B(t)=B^{r}(\tau(t)), \quad \quad t\geq0,
    $$
    where $B^{r}$ is a Brownian motion starting at $r$ and reflected at zero,  
    and $\tau(t) = \inf\{s \geq 0 : \C(s) > t\}$.
\end{definition}
Note that this defines a process restricted to the  
positive half of the real line. If $\xi = 0$, the expression  
for the clock process in \cref{thm:def_clock_time}  
reduces to that of the ``quantum clock"  
introduced in \citet{garban2016liouville},  
albeit in the one-dimensional setting.  
In this case, we recover the  
one-dimensional LBM, also restricted to the positive half of the real line.
We also stress that the existence of Brownian local times  
is a particular feature of the one-dimensional setting,  
which allows us to work directly with the unregularized process.  
This greatly simplifies the LBM theory developed in  
\citet{garban2016liouville}, avoiding the need to discuss  
convergence of approximations.

Using the notation introduced in \cref{sec:direct},  
and in the same vein as \cref{thm:mfk is markov}, a direct application of Theorem 16.56 
in \cite{breiman1992probability} leads to 
\begin{theorem}\label{thm:MLBM is markov}
    $\P^\gf$-a.s.,
    under $\P^B_{r}$,
    the MLBM
    is a stationary strong Markov process
    with continuous sample paths and speed measure given by

    \begin{equation}\label{eq:speed5}
        \frak M(dr)
    = \mathds{1}_{(0,\infty)}(r)\,\frak m(dr)
    + \frak M_0 \delta(dr),\quad \text{with}\quad \frak{m}(dr)=r^{-\xi}\mu(dr).
    \end{equation}
\end{theorem}

For $\xi= 0$, corresponding to the LBM, the absolutely continuous component of the speed measure reduces to the GMC. 
In this sense, the LBM is the canonical diffusion associated to  the random geometry induced by the GMC.
\subsection{Phases of the 1D MLBM}
\subsubsection{Quenched setting}
\cref{thm:MLBM is markov} characterizes the speed measure  
of the MLBM, allowing us to identify its phases in a similar manner  
to the multifractal Kraichnan diffusion, both in the quenched and the annealed setting. 
To proceed, we need one  computational lemma  concerning the integrability of the GMC, 
which can be found in the lecture notes \cite{rhodes2016lecture}.
\begin{lemma}[Seiberg bound]\label{thm:integrability}
    Let $\xi \in \mathbb{R}$ and  $\delta>0$.
    Then, 
    \be
    \int_{0}^{\delta}r^{-\xi}\gmc(dr) <  +\infty,
    \quad \mathrm{a.s.}
    \ee
    if and only if 
    $\xi < 1+2\gamma^2$.
\end{lemma}

The boundary behavior at $r=0$ of the 1D MLBM is then summarized by the 
\begin{theorem}\label{thm:phases MLBM}
    Let $\xi\in [0,2]$ and $\gamma \in \big[0,\frac{\sqrt{2}}{2}\big)$.
    Then, for the MLBM, the boundary point $r=0$ is,
    $\,\,\P^\gf$-a.s.,\\
%    \begin{itemize}
	\be
	\nonumber
        \text{(I)  regular, if }\,\,\,\xi< 1 + 2\gamma^2,\quad
        \text{(II)  exit, if }\,\,\,1 + 2\gamma^2\leq
        \xi< 2 + 2\gamma^2,\quad
        \text{(III) natural, if }\,\,\,2 + 2\gamma^2\leq\xi.
%    \end{itemize}
	\ee
    \noindent The corresponding phase diagram is displayed in the left panel of Fig.~\ref{fig:MLBM-phases}.
\end{theorem}
\begin{proof}
  From \cref{thm:closed_point},  we first observe that 
    the origin is an accessible boundary point
    if and only if 
    $\int_{0}^{\delta}r^{1-\xi}\gmc(dr) < \infty$,
    for any $\delta>0$. 
    This holds true if and only if $\xi<2 + 2\gamma^2$,
    according to Lemma \ref{thm:integrability}. 
   \cref{thm:natural} then yields phase (III).
   When $\xi<2 + 2\gamma^2$, $0$ is accessible.
   The distinction between (I) and (II) is determined by the property \eqref{eq:exit}, stating that
    0 is regular if and only if 
    $\int_{0}^{\delta}r^{-\xi}\gmc(dr) < \infty$,
    for all $\delta>0$.
    From Lemma \ref{thm:integrability},
    this is true if and only if  
    $\xi<1 + 2\gamma^2$.
 \end{proof}
\subsubsection{Annealed setting}
An annealed version of the MLBM can be defined upon replacing
$\mu(dr) \mapsto \mathbb E^\Gamma \mu(dr) =dr$ in \cref{thm:def_clock_time},
leading to $\mathfrak m \mapsto \mathbb E^\Gamma \frak m(dr) = r^{-\xi}dr$ in \cref{thm:MLBM is markov}. 
This (annealed) averaging over the environment recovers the effective  process of Eq.~\eqref{eq:MBM} accounting for the 
 monofractal Kraichan case. In short, the annealed MLBM describes the separation process in (monofractal) Kraichnan, with boundary behavior prescribed by \cref{thm:phase}.

\begin{figure}[h]
	\includegraphics[width=0.49\textwidth]{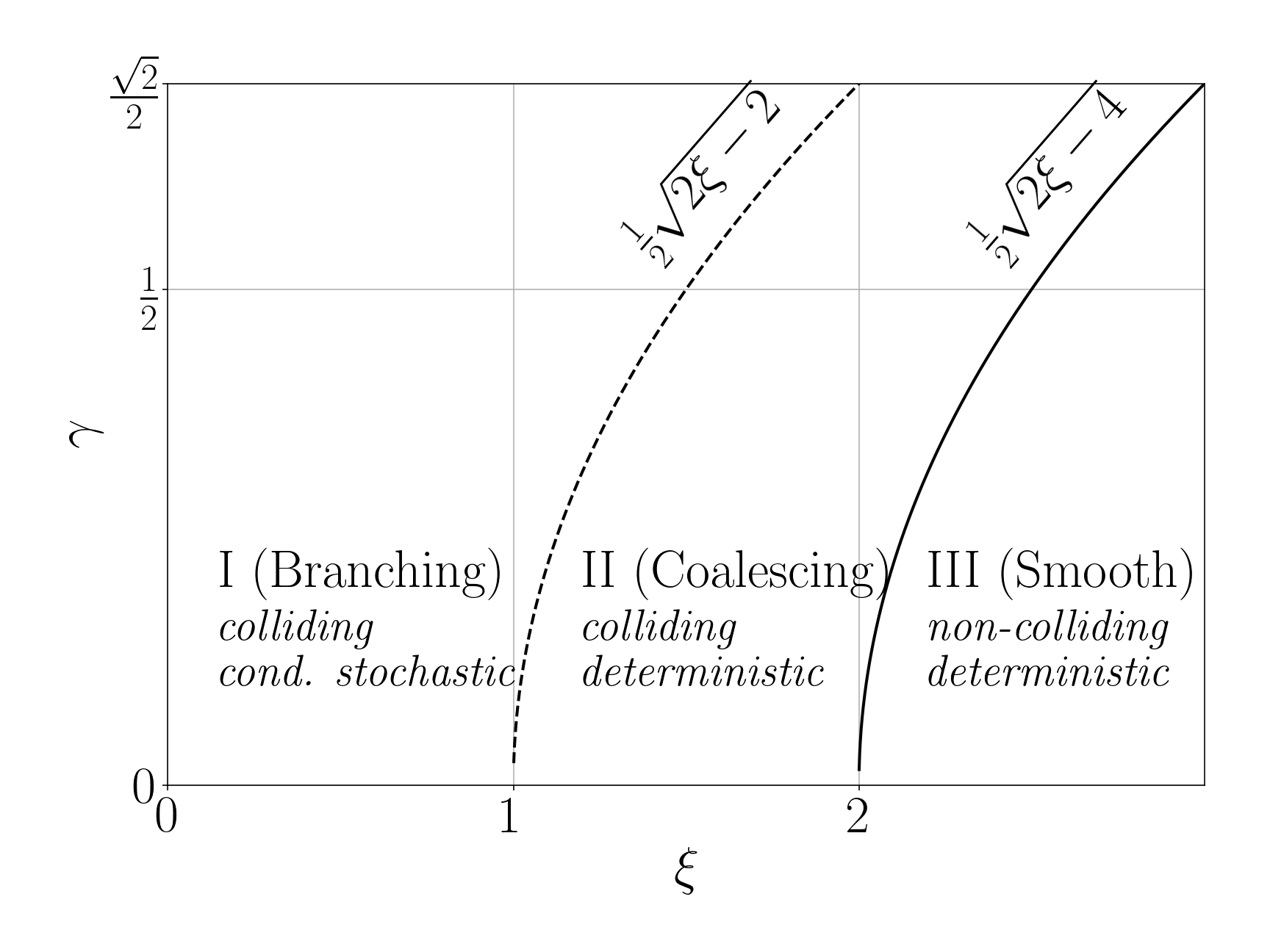}
	\includegraphics[width=0.49\textwidth]{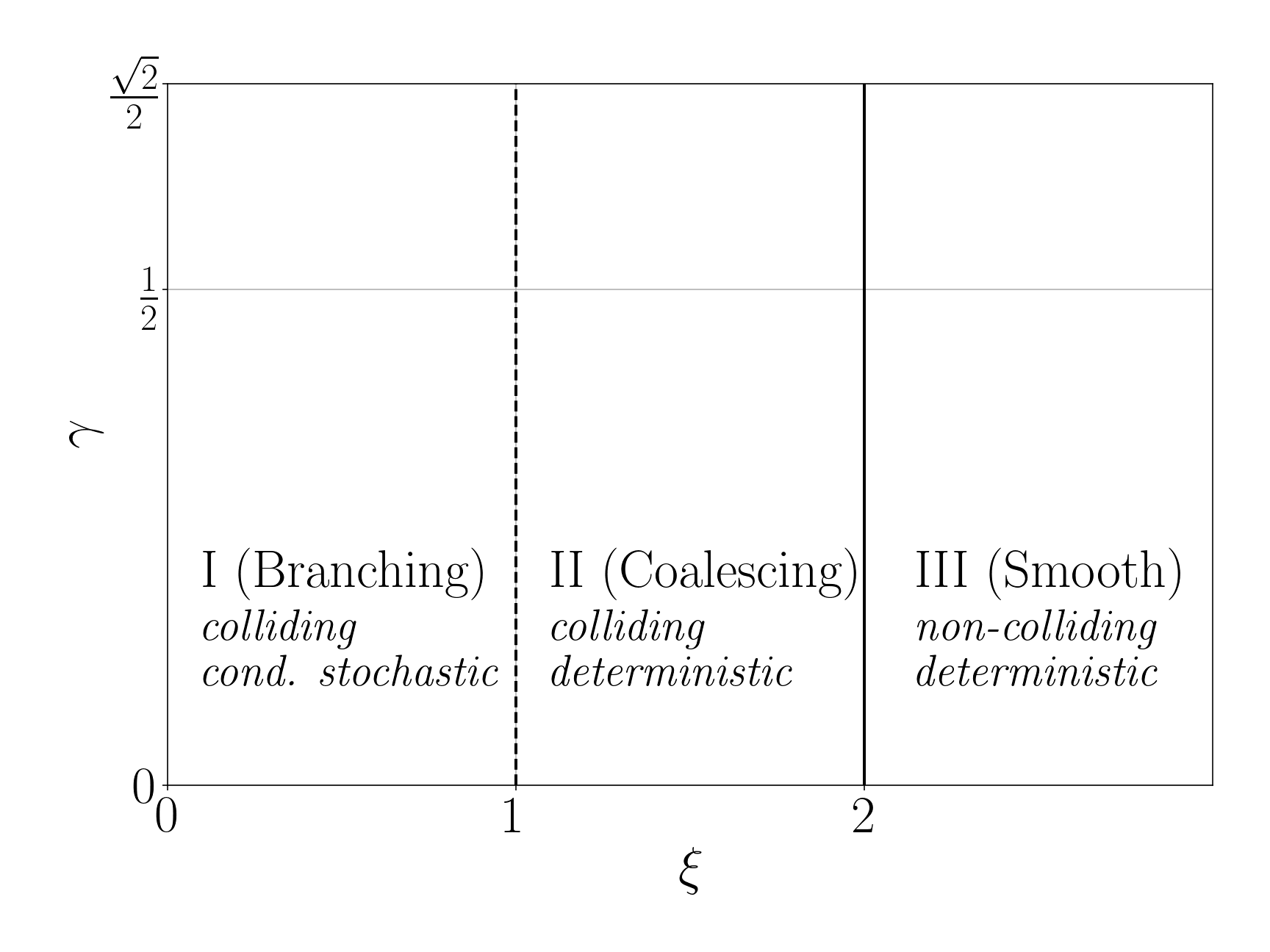}
	\caption{Phases of the Lagrangian flow in the MLBM model. Left: Quenched setting. Right: Annealed setting.}
	\label{fig:MLBM-phases}
	\end{figure}

\subsection{MLBM vs Multifractal Kraichnan (MK)}
Contrasting the separation process in multifractal Kraichnan (MK) to the MLBM, the intermittency induced by the GMC  gives rise to apparently very different effects, both at the level of the phase diagrams and for the multifractal properties of the trajectories.
%In particular, the MLBM phase diagrams of Fig.~\ref{fig:MLBM-phases} indicate  a roughening by intermittency instead of the smoothening effect of \cref{sec:mfk-phases}. This is reflected by different multifractal behaviors at the level of the trajectories: the MK process is driven by typical points  while MLBM is  driven by the $2-\gamma$ thick points.
Remarkably, though, we point out a mapping between  the MK and MLBM phase diagrams (Figs. \ref{fig:phases} and \ref{fig:MLBM-phases}), 
valid for both the quenched and annealed setting.

\subsubsection{Roughening by intermittency }
When intermittency is absent ($\gamma = 0$),  
the multifractal Kraichnan model and the MLBM coincide,  
describing the same multiplicative diffusion process,  
with quadratic variation proportional to $r^\xi$ --- this corresponds to the inter-particle separation  
in the one-dimensional Kraichnan model,  
where the driving parameter is the (deterministic) structure function of the Kraichnan velocity field.
For $\gamma > 0$, the two models give rise to  
fundamentally different processes. 
While intermittency has a smoothing effect   on Lagrangian dispersion in the multifractal Kraichnan model,  
it leads to rougher behavior in the MLBM,   amplifying the explosive nature of trajectories.
This contrast can be quantified by comparing the phase diagrams  
of the two models, shown in Figs.~\ref{fig:phases}  and \ref{fig:MLBM-phases}, respectively.
In the multifractal Kraichnan model, increasing $\gamma$ causes  
the regular region (I) to shrink, while the exit region (II) expands.  
As a result, the system undergoes a transition from phase I to phase II  
as $\gamma$ increases, marking a shift from explosive, spontaneously stochastic behavior,  
to a regime in which trajectories coalesce.
In the MLBM, on the other hand, we observe the opposite effect:   the regular region grows and occupies a larger portion of the phase space. 
Consequently, the phase transition now occurs in the opposite direction,  
from phase II to phase I, indicating that intermittency enhances  
the explosive nature of particle behavior in the MLBM.

\subsubsection{Multifractal behavior}
\label{sssec:MFbehavior}
The difference in the phase diagrams between the MLBM and the MK reflects different couplings with the underlying thick points defined in \cref{ssec:thickpoints}.
This difference comes from the fact that far from boundary, the MK speed measure is absolutely continuous with the Lebesgue mesure, while the MLBM speed measure is not--- and prescribed by the GMC.
%In this subsection, we analyze the behavior of the MLBM  
%with respect to thick points of $\gf$, in a similar manner  
%to \cref{sec:mfk multifractal behavior}.  
%For simplicity, we restrict the MLBM to the  
%bounded interval $I = [0,1]$,  
%denoting its interior by $\mathring{I}$.  
%To facilitate the analysis, we also introduce the hitting times 
%$T_0=\inf\left\{t\geq 0 \,\,\,;\,\,\, \B_t=0   \right\}$,
%$T_1=\inf\left\{t\geq 0 \,\,\,;\,\,\, \B_t=1   \right\}$,
%and set $T=T_0 \wedge T_1$, so that $T$ represents the first exit time of
%$\B$ from $I$.
As a consequence, the MLBM  spends almost all its time on  the $2\gamma$-thick points of $\gf$, or in other words, the most singular structures give mass to the GMC. This is stated by the following
\begin{proposition}\label{thm:MLBM spends time}
    The MLBM
    spends Lebesgue-almost all its time in $(2\gamma)$-thick points of $\,\,\gf$,
    that is,
    $$
    \P^\gf \text{-} a.s., \;\; \text{for all}\;\; r\in (0,L), \;\;
    \P^B_r\text{-}a.s.,\;\; 
    \text{Leb}\left(\left\{0\leq t\leq T \,\,:\,\, \B(t)\in \T_{2\gamma}^c \right\}\right)=0,
    $$
    where $T$ is the first exit time of $\B$ from $(0,L)$, and $\T_{\alpha}$ is defined in \eqref{eq:thickpoints}.
\end{proposition}
The proof is similar to that of \cref{thm: mfk spends time}. For the 2D LBM, this result was shown in \citet{rhodes2014heat}. 
For $\gamma>0$, the $2\gamma$-thick points form a fractal set of zero Lebesgue measure  
and are rare and exceptionally singular  from a probabilistic standpoint, reflecting a fundamental difference between the MK and the MLBM.
%This difference can be explained  
%by examining the speed measures of the two processes.  
%Since the time spent in a given region is proportional  
%to the value of the speed measure on that region,  
%each process tends to concentrate  
%on the carrier set of its respective speed measure.  
%For the multifractal Kraichnan process, the carrier  
%is easily seen to be the set of Lebesgue-typical points,  
%while the speed measure of the MLBM is carried  
%by the set $\T_{2\gamma}$ of $(2\gamma)$-thick points.

\subsubsection{MK as a MLBM with parameter $(\xi+4\gamma^2,\gamma)$}
The MK and the MLBM are prescribed in terms of two parameters:  
$\xi$, which characterizes the behavior near the origin,  
and $\gamma$, which controls the level of intermittency.  
Although they exhibit different phase transitions for a given  
pair $(\xi, \gamma)$, the phase transitions of the multifractal Kraichnan model  
can be derived from those of the MLBM  
by performing a suitable change of parameters, valid for both the annealed and the quenched setting.
By applying the mapping $\xi \mapsto \xi + 4\gamma^2$  
in \cref{thm:phases MLBM}, one recovers exactly the phase transitions  
of the multifractal Kraichnan model, 
as stated in \cref{thm:multifractal phases}.
Under this mapping, the underlying Hölder exponent  
increases, and particle motion becomes smoother.  
Consequently, the rougher dynamics of the MLBM  
is transformed into the smoother dynamics  
of the multifractal Kraichnan model.  
In other words, the quenched and annealed phase diagrams of the MLBM with parameters  
$(\xi + 4\gamma^2, \gamma)$  
are the same as that of the multifractal Kraichnan model  
with parameters $(\xi, \gamma)$.
The transformation $\xi \mapsto \xi - 4\gamma^2$  
maps the multifractal Kraichnan model into the MLBM  
in a similar fashion.
In this regard, the correspondence between the multifractal Kraichnan  
model and the MLBM mirrors the relationship between the monofractal  
Kraichnan model and the Bessel process (see \cref{ssec:Besselshortcut}).  
The MLBM emerges as an effective model for inter-particle separation  
in multifractal Kraichnan flows, allowing the phase diagram of the latter  
to be retrieved from the former, while replacing the dynamics  
based on the velocity field \eqref{eq:velocity}  
with an effective diffusion process, only coupled to a random geometry induced by the GMC.
We point out, though, that this correspondance holds at the level of the phase diagram only and does not alter the multifractal coupling with the thick points discussed in \cref{sssec:MFbehavior}.

%%%%%%%%%%%%%%%%%%%%%%%%%%%%%%%%%%%%%%%%%
%%%%%%%% 7: CONCLUDING REMARKS
%%%%%%%%%%%%%%%%%%%%%%%%%%%%%%%%%%%%%%%%%
\section{Concluding remarks}\label{sec:remarks}

In this work, we have investigated the signature of Eulerian intermittency in the statistics of Lagrangian dispersion by analyzing  
a solvable unidimensional multifractal 
extension of the Kraichnan model. 
It is  obtained by imposing a 
white-in-time structure to a 1D version of the multifractal hydrodynamic random field constructed by \cite{robert2008hydrodynamic}, and considering transport in a quasi-Lagrangian setting. This  setting  is a quasi-Lagrangian version of the model previously  investigated with numerics and mean-field heuristics in \cite{considera2023spontaneous}.
It provides a  minimal solvable model  showcasing both the physical effects and the mathematical challenges of incorporating spatial intermittency into the Kraichnan model.
Our main  result is Theorem \ref{thm:multifractal phases} --- 
summarized in Fig.~\ref{fig:phases} --- which formulates a  \emph{smoothing-by-intermittency} effect at the level of the Lagrangian flow.  
Besides, we discussed analogies between the separation process in the
multifractal Kraichnan (MK) model and the theory of Liouville Brownian motion (LBM). 
In particular, we argued that the phase diagram is recovered by a multiplicative version of LBM with parameter $(\xi+4\gamma^2,\xi)$.

The duality between the MK model  and LBM highlights a connection between turbulent transport and diffusion processes in Liouville Quantum Gravity (LQG). 
Let us however observe that  although both turbulence and LQG  
make use of GMC theory to model intermittency,  its role in each context is fundamentally different.
In LQG, intermittency is tied to the deformation of the geometry itself.  
The GMC modifies distances and alters the underlying metric,  
so that diffusion, as described by the MLBM,  takes place within a highly irregular random geometry.
In turbulence, by contrast, intermittency arises from strong fluctuations  in the velocity field and affects how particles are transported.  
Here, the GMC encodes spatial variations in velocity intensity  and governs the dynamics of the MK process.  
This distinction reveals that although GMC is the common  mathematical foundation, its effect on diffusion is mediated  
by the physical nature of the model. In short: in LQG, intermittency   shapes the geometry through which diffusion unfolds,  
whereas in turbulence, it modifies the driving velocity field.   As a result, the multifractal properties of the associated  diffusion processes differ in essential ways.
This perspective naturally invites a reinterpretation  
of intermittency in turbulence as a manifestation  
of a random geometry in its own right.  
Although intermittency in turbulence is not traditionally  
formulated in geometric terms, as it is in the context of LQG,  
the multifractal structure of the velocity field
exhibits features that closely parallel those  
of a random geometric landscape.  This connection may prove of mathematical and physical relevance for modeling turbulent transport:
One could envision a ``turbulent geometry'' in which transport processes  unfold, 
and  speculate that other geometric notions of relevance in LQG --- such as metric tensor, curvature, or geodesics --- have meaningful counterparts in the context of turbulent transport.

Among natural continuations of this work, 
one may wish to increase further the flow 
complexity: Our analysis was conducted in a one-dimensional setting 
and a similar investigation for higher-dimensional models 
remains to be carried out.  The multidimensional case, in particular, 
would allow for an investigation of the  effects of the compressibility degree. 
In this regard, a naive guess suggests that the phase transitions of multidimensional flows 
can be obtained from the monofractal Kraichnan case  
by substituting the roughness exponents $\xi$ 
with their most probable counterpart $\xi_{\text{eff}} = \xi + 2\gamma^2$ and preliminary numerical simulations suggest that this view is essentially correct. 
Another way to increase complexity is to allow for correlations between the GMC and the Brownian sheet entering the RV velocity field. This complication turns out to be  relevant for turbulence modeling as it allows to model the skewness phenomenon \cite{chevillard2019skewed}.

Another research direction is to address multiparticle ($N\ge 3$) dynamics, and implications on scalar transport.
While here we focused on two-particle motion,  it is natural to expect Eulerian intermittency to leave a signature on the geometry of multiparticle Lagrangian clusters.  One difficulty is that \cb the geometric description of a cluster with $N\ge 3$ particles does not  reduce to a unidimensional diffusion process, and requires different stochastic tools other than the speed measure characterization of the underlying process. 
Finally, one may want to bridge the insights of  multifractal Kraichnan flows to turbulent data analysis. Our quenched setting describes a conditioning on typical realizations of Gaussian multiplicative chaos and  extending our analysis to atypical realizations could be instructive. In particular, and from the point of view of turbulence phenomenology, conditioning on Gaussian multiplicative chaos   corresponds to conditioning with respect to the dissipation field. This approach could potentially be applied to the analysis of turbulent data as a mean of extracting scaling laws that might otherwise remain hidden.
\cb

\paragraph{Acknowledgements.}
We thank A. Barlet, J. Bec, A. Cheminet, B. Dubrulle,  A. A. Mailybaev, E. Simonnet \& N. Valade  for continuing discussions. ST acknowledges
support from the French-Brazilian network in Mathematics for  several Southern summer visits at \emph{Instituto de Matem\'atica Pura e Aplicada (IMPA)} and thanks IMPA for support and hospitality during those stays. AC acknowledges support from the \emph{Académie d'Excellence Systèmes Complexes} within  Université Côte d'Azur  for two extended research stays at \emph{Institut de Physique de Nice}. 
%%%%%%%%%%%%%%%%%%%%%%%%%%%%%%%%%%%%%%%%%
%%%%%%%% APPENDICES
%%%%%%%%%%%%%%%%%%%%%%%%%%%%%%%%%%%%%%%%%
\appendix
\section{Collected proofs}
\label{sec:proofs}
In this appendix, we provide the proofs of 
\cref{thm:A bound}, \cref{thm:power-law spectrum A}, 
\cref{prop:Aestimate} and \cref{thm: mfk spends time},
along with the statements of auxiliary lemmas.

\subsection{Proof of \cref{thm:A bound}}
\begin{proof}[Proof of \cref{thm:A bound}]
Without loss of generality, we assume $L=1$ for simplicity. 
The starting point is the expression for the unregularized 
diffusion coefficient $\mA$, obtained by setting $\eta=\kappa=0$ in 
Eq.~\eqref{eq:coefficient kraichnan}.
It is explicitly written as 
\be
	 \mA(r) =  \dfrac{1}{2} \int_\R dz\,\left( \psi(r-z)\dfrac{r-z}{|r-z|^{3/2-\xi/2}}-\psi(-z)\dfrac{-z}{|z|^{3/2-\xi/2}}\right)^2.
\ee
Upon applying the change of variables $z \mapsto r z$, it may be further expressed as 
\be
	\label{eq:cd-bis}
	 \mA(r) = C(r)|r|^\xi, \quad \text{with} \quad C(r) :=  \dfrac{1}{2} \int_\R dz\,\left(  \psi(r(1-z))\dfrac{1-z}{|1-z|^{3/2-\xi/2}}-\psi(-r z)\dfrac{z}{|z|^{3/2-\xi/2}}\right)^2.
\ee
Note that $C(r)$ is strictly positive for $r\in[-1,1]$, 
and continuous by dominated convergence.
The claim then follows by setting 
$C_+ = \sup_{r \in [-1,1]} C(r)$
and 
$C_- = \inf_{r \in [-1,1]} C(r)$, 
and by observing that both $C_+$ and $C_-$ 
are finite and strictly positive by the 
extreme value theorem.
\end{proof}

\subsection{Proof of \cref{thm:power-law spectrum A}}
Let us first state three lemmas  due to \citet{robert2008hydrodynamic}, here  adapted to our notations.
\begin{lemma}
    \cite[Proposition 3.1]{robert2008hydrodynamic}
    Given a positive integer $m$ and $f:\R \to \R$ satisfying
    $$ \int_{\R^{2m}} \left|f(z_1)\right| \cdots \left|f(z_{2m})\right| 
    \left(\prod_{1\le i<j \le 2m} \dfrac{1}{\left(\left|\frac{z_i-z_j}{L} \right| \wedge 1\right)^{4\gamma^2}} \right)
    dz_1\cdots dz_{2m}<+\infty,$$
    we have, for $p\leq 2m$,
    \be
        \E\left[ \left(\int_\R  f(y) \mu(dy) \right)^p\right] = \int_{\R^p} f(z_1)\cdots f(z_p) \left(\prod_{1\le i<j \le p} \dfrac{1}{\left(\left|\frac{z_i-z_j}{L} \right| \wedge 1\right)^{4\gamma^2}} \right) dz_1\cdots dz_{p}.
    \ee
\end{lemma}

\begin{lemma}[Product bound]
	\cite[Lemma 2.4]{robert2008hydrodynamic}
	\label{lemma:productbound}
	Given $\sigma$ a finite positive measure on $\R$,
	$q: \R\times \R \to \R_+$ a symmetric application 	and $m$ a positive integer, the following  inequalities hold
	\be
	\int_{\R^{p}} e^{\sum_{1\le j<k\le p} q(z_j,z_k)} \sigma(dz_1)\cdots\sigma(dz_{p}) \le
	\begin{cases}
	&  \sigma(\R) \left(\sup_{x\in \R}\int_\R e^{mq(x,z)}\sigma(dz) \right)^{2m-1} \hspace{1cm} (p=2m)\\
	&  \sigma(\R) \sup_{x,y\in \R} \left(\int_\R e^{q(y,z)}\sigma(dz) \right) \left(\int_\R e^{mq(x,z)+q(y,z)}\sigma(dz) \right)^{2m-1} \hspace{0.5cm} (p=2m+1)
	\end{cases}
	\ee	
\end{lemma}

\begin{lemma}[Supremum scaling]
    \cite[Lemma 3.4]{robert2008hydrodynamic}
    \label{lemma:supbound}
    For all $\delta \in [0,\xi)$, there exists $C_\delta>0$ such that, for all $r \in[-L,L]$, the following inequality holds
    \be
    	\sup_{x\in \R}\int_\R \dfrac{\left(\Delta_r\phi(z)\right)^2}{\left(\frac{|x-z|}{L} \wedge 1\right)^\delta} dz \le C_\delta \left(\frac{|r|}{L}\right)^{{\xi-\delta}},
    \ee
    where $\Delta_r\phi(z):=\phi(r-z)-\phi(-z)$
\end{lemma}

We can now proceed to the
\begin{proof}[Proof of \cref{thm:power-law spectrum A}]
     Without loss of generality we assume $L=1$. 
    In the case $p=1$, the inequality is a consequence of the property $\E \mu(dz) = dz$ and the estimate 
    \be
        \E^\Gamma  \mA^{(\gamma)}(r) =  \int_\R \left[\phi(r-z) - \phi(-z)\right]^2 dz.
    \ee
    From Eq.~\eqref{eq:cd}, this implies in particular that $\E^\Gamma  \mA^{(\gamma)}(r) \approx r^\xi$ for $ r \in (0,1)$.
    
    Now consider $p$ a positive integer such that $1<p < \frac{2\xi}{\gamma^2}$.
     We recall the notation $\Delta_r \phi(z):=\phi(r-z)-\phi(-z)$  and introduce the shorthand $|\cdot|_\ast = |\cdot|\wedge 1$.
    Combining the definition \eqref{eq:coefficient kraichnan} with $\eta=\kappa=0$ yielding  the unregularized diffusion coefficient
    $\mA^{(\gamma)}$ and  the integrability property \ref{thm:integrability} on moments yields the following expression for the moment of order $p$:
        $$
        \E^\gf\left(\left(\mA^{(\gamma)}(r)\right)^{p}\right)
        = \int_{\R^p}
        \left(\Delta_r\phi(z_1) \dots\Delta_r\phi(z_p)\right)^2
        \prod_{1\leq i<j\leq p} \frac{1}{\left|z_i-z_j\right|_\ast^{\gamma^2}}
        dz_1\dots dz_p. \nonumber
        $$
        Next, we apply the product bound lemma \ref{lemma:productbound}  with the measure $\sigma_r(dz)=(\Delta_r\phi(z))^2 dz$ and $q(x,z)=-\gamma^2 \log_+|x-z|$. From Eq.~\eqref{eq:cd}, it is immediate to    verify that $\sigma_r$ is a finite positive measure on $\R$ with total mass $\int_\R \sigma_r(dz)\approx |r|^\xi$. This yields the bounds
    \be
    \label{eq:Asup}
    \E^\gf\left(\left(\mA^{(\gamma)}(r)\right)^{p}\right)\leq 
        \begin{cases}    
            & |r|^\xi
            \left(\sup_{x\in\R}\int_{\R}
                \frac{\left(\Delta_r\phi(z)\right)^2 \, dz}{\left|x-z\right|_\ast^{m\gamma^2}}
            \right)^{2m-1} \hspace{0.5cm} (p=2m)\\
          & |r|^\xi
    \sup_{x,x'\in\R}
            \int_{\R}  \frac{\left(\Delta_r\phi(z)\right)^2 \,dz}{\left|x'-z\right|_\ast^{\gamma^2}}
                \left(
            \int_{\R}
                    \frac{ \left(\Delta_r\phi(z)\right)^2 \, dz}{\left|x-z\right|_\ast^{m\gamma^2}\left|x'-z\right|_\ast^{\gamma^2}}
            \right)^{2m-1} \hspace{0.5cm} (p=2m+1)
        \end{cases}
    \ee

    Let us now discuss the scaling along even and odd integers separately.
    For even integer $p=2m$
    we can directly apply the supremum scaling lemma \ref{lemma:supbound}
    with $\delta= m\gamma^2 <\xi$, to find a constant $C_p>0$ such that for all $r\in (0,1)$
        \be
        \label{eq:oddb1}
           \E^\gf\left(\left(\mA^{(\gamma)}(r)\right)^{p}\right) \leq C_p |r|^{\xi +(\xi -m\gamma^2)(2m-1)}= C_p|r|^{\zeta_\mA(2m)}
        \ee
    and this provides the desired scaling.
    
    For odd  integer $1<p=2m+1$, we estimate, using sucessively the Cauchy-Schwartz inequality and the supremum scaling lemma  \ref{lemma:supbound}
    \be
        \label{eq:oddb2}
             \int_{\R}
                    \frac{ \left(\Delta_r\phi(z)\right)^2 \, dz}{\left|x-z\right|_\ast^{m\gamma^2}\left|x'-z\right|_\ast^{\gamma^2}} \le
            \left(\int_{\R}
                    \frac{ \left(\Delta_r\phi(z)\right)^2 \, dz}{\left|x-z\right|_\ast^{2m\gamma^2}}
            \int_{\R}
                    \frac{ \left(\Delta_r\phi(z)\right)^2 \, dz}{\left|x'-z\right|_\ast^{2\gamma^2}}\right)^{1/2}\le (C_{2m}C_2)^{1/2}r^{\xi-(m+1)\gamma^2}
    \ee
    We then use the supremum scaling lemma one more time to bound the remaining term in the rhs of Eq.~\eqref{eq:Asup}. This yields
        \be
           \E^\gf\left(\left(\mA^{(\gamma)}(r)\right)^{p}\right) \leq \tilde C_p |r|^{2\xi-\gamma^2+(\xi-(m+1)\gamma^2)(2m-1)}= \tilde C_p|r|^{\zeta_\mA(2m+1)}
        \ee
    Note that the use of  supremum scaling lemma in Eq.~\eqref{eq:oddb1}-\eqref{eq:oddb2} requires to check the inequality $\gamma^2 < \max(\xi,\xi/2,\xi/2m)$. It is indeed verified from the bounds on $p \in  [3 , 2\xi/\gamma^2)$.
\end{proof}
    \subsection{Proof of \cref{prop:Aestimate}}
\begin{proof}[Proof of \cref{prop:Aestimate}]
    Let us first consider the dyadic sequence $r_n=2^{-n}$, for $n\geq 1$.
    Markov's inequality and  the \cref{thm:assumption} yield
    \begin{align}
        \P^\gf\left(\mA^{(\gamma)}(r_n)>r_n^{(\xi+2\gamma^2-\eps)}\right)
        &\leq\frac{1}{2^{-np(\xi+2\gamma^2-\eps)}} \E^\gf\left((\mA^{(\gamma)}(2^{-n}))^p\right)\\
        &=\frac{C}{2^{n(-p(\xi+2\gamma^2-\eps)+\zeta_\mA(p))}},
    \end{align}
    for arbitrarily small $p$.
    In order for the sequence above to be summable,
    we must have $\zeta_\mA(p)/p>(\xi+2\gamma^2-\eps)$.
    Observe that $\zeta_\mA^{\prime}(p)|_{p=0}=\xi+2\gamma^2$,
    therefore we can find $p$ small such that 
    $\zeta_\mA(p)/p>(\xi+2\gamma^2-\eps)$. This ensures that the sequence above is 
    summable. Thus, we can use Borel-Cantelli lemma to conclude that\;
    $\P^\gf$-a.s., \,
    $\displaystyle
    \mA^{(\gamma)}(r_n)\leq  r_n^{\xi+2\gamma^2-\eps},
    $
    for $n$ large enough. 
    The general case follows by considering a general sequence $r_n$ going to $0$
    and approximating it by the dyadic sequence above.

    To deduce the lower bound we follow a similar approach. 
    Markov's inequality and \cref{thm:assumption} lead to
    \begin{align}
        \P^\gf\left(\mA^{(\gamma)}(r_n)<r_n^{(\xi+2\gamma^2+\eps)}\right)&=
        \P^\gf\left(\mA^{(\gamma)}(r_n)^{-p}>r_n^{-p(\xi+2\gamma^2+\eps)}\right)\\
        &\leq\frac{1}{2^{np(\xi+2\gamma^2+\eps)}} \E^\gf\left((\mA^{(\gamma)}(2^{-n}))^{-p}\right)\\
        &=\frac{C}{2^{n(p(\xi+2\gamma^2+\eps)+\zeta_\mA(-p))}}
    \end{align}
    for arbitrarily small $p$. Along the same lines as before, 
    we can find $p$ small such that the sequence above is summable.
    Applying Borel-Cantelli lemma, we deduce that $\P^\gf$-a.s.\;
    $\displaystyle
    \mA^{(\gamma)}(r_n)\geq  r_n^{\xi+2\gamma^2+\eps},
    $
    for $n$ large enough. The general case also follows by approximation.
\end{proof}

\subsection{Proof of \cref{thm: mfk spends time}}
We first state the fundamental lemma characterizing the speed measure
of the multifractal Kraichnan process, which is an immediate consequence 
of \cite{revuz2013continuous}, Corollary VII.3.8. 
Recall that $T$ denotes the first exit time of $R$ out of $(0,L)$.

\begin{lemma} \label{thm: occupation lemma}
    Let $r \in (0,L)$ and let $f$ be a positive  
    Borel function on $[0,L]$. Then, $\P^\gf$-a.s., the following formula holds:
    \begin{equation}\label{eq:green function}
        \mathbb{E}^B_r \left[\int_0^{T} f(R(s)) \, ds \right] = \int_{0}^L G(r,r') f(r') \frak m^{(\gamma)}(dr'),
    \end{equation}
    where $G$ is the Green function for the one-dimensional Laplacian  
    on $[0,L]$ with (vanishing) Dirichlet boundary conditions.  
\end{lemma}

Observe that by taking the indicator function $\ind_E$ of  
some Borel set $E \subset [0,L]$, the left-hand side of  
\eqref{eq:green function} corresponds to the expectation of the  
total time spent by $R$ in the set $E$, up to the exit time $T$, recovering Eq.~\eqref{eq:escape}.
Thus, the speed measure $\frak m^{(\gamma)}$ can indeed be viewed as a measure  
of the time spent by the multifractal Kraichnan process.  
We can now proceed with  the
\begin{proof}[Proof of  \cref{thm: mfk spends time}]
    Since $\frak m^{(\gamma)}$ has a density with respect to the Lebesgue measure 
    in $(0,L)$ , the set of typical points also gives full mass to $\frak m^{(\gamma)}$,
    meaning that, $\P^{\gf}$-a.s., $\frak m^{(\gamma)}(\T_0^c)=0$.
    Thus, by applying \cref{thm: occupation lemma} with $f=\ind_{\T_0^c}$, 
    we deduce that $\P^\gf$-a.s. for all $r\in (0,L)$, 
    $$
    \mathbb{E}^B_r \left[\int_0^{T} \ind_{\T_0^c}(R(s)) \, ds \right] = 0.
    $$
    Since the random variable inside the expectation is non-negative, 
    it follows that $\P^B_r$-a.s., 
    $$
    \int_0^{T} \ind_{\T_0^c}(R(s)) \, ds=0.
    $$
    Similarly, we deduce that
    $\ind_{\T_0^c}(R(s))=0$
    for Lebesgue-almost every $s\in [0,T]$.
    We conclude the proof by observing that this is equivalent 
    to the set in the statement of the proposition having vanishing 
    Lebesgue measure.
\end{proof}

\bibliography{biblio}% common bib file
\end{document}